\documentclass[smallextended,envcountsect]{svjour3}       
\smartqed  

\usepackage{amsmath, amsfonts, amssymb,  mathrsfs,  array, stmaryrd,  indentfirst}
\usepackage{color}
\usepackage[normalem]{ulem}
\usepackage[colorlinks,linkcolor=blue,citecolor=blue]{hyperref}

\journalname{MFE}

\numberwithin{equation}{section}
\newtheorem{assumption}{Assumption}[section]

\newcommand{\cadlag}{c\`{a}dl\`{a}g }
\newcommand{\ladlag}{l\`{a}dl\`{a}g }

\newcommand{\essinf}{{\mathrm{ess}\inf}}
\allowdisplaybreaks

\begin{document}

\title{Optimal Investment with Random Endowments and Transaction Costs: Duality Theory and Shadow Prices}

\titlerunning{Duality Theory and Shadow Prices}        

\author{Erhan Bayraktar        \and
        Xiang Yu
}

\authorrunning{E. Bayraktar and X. Yu} 

\institute{E. Bayraktar \at
              Department of Mathematics, University of Michigan, 530 Church Street, Ann Arbor, MI 48109, USA \\
              \email{erhan@umich.edu}
           \and
           X. Yu \at
              Department of Applied Mathematics, The Hong Kong Polytechnic University, Hung Hom, Kowloon, Hong Kong\\
              \email{xiang.yu@polyu.edu.hk}
}

\date{Received: date / Accepted: date}

\maketitle

\begin{abstract}
This paper studies the utility maximization on the terminal wealth with random endowments and proportional transaction costs. To deal with unbounded random payoffs from some illiquid claims, we propose to work with the acceptable portfolios defined via the consistent price system (CPS) such that the liquidation value processes stay above some stochastic thresholds. In the market consisting of one riskless bond and one risky asset, we obtain a type of super-hedging result. Based on this characterization of the primal space, the existence and uniqueness of the optimal solution for the utility maximization problem are established using the duality approach. As an important application of the duality theorem, we provide some sufficient conditions for the existence of a shadow price process with random endowments in a generalized form similar to \cite{Chris22} as well as in the usual sense using acceptable portfolios.
\end{abstract}

\section{Introduction}

The optimal investment via utility maximization is a fundamental research topic in quantitative finance. In frictionless markets, the problem with both liquid assets and illiquid contingent claims has recently received much attention and has been significantly developed. It is assumed in the model that the investor receives random payoffs from some contingent claims at the terminal time $T$. In complete markets, random endowments can be perfectly hedged using a dynamic trading portfolio with liquid assets. As a consequence, the optimal investment problem with some payoffs reduces to the one without random endowments, but with an augmented initial wealth. When the market is incomplete, the problem becomes more delicate. In particular, to build the convex duality theorem treating unhedgeable random endowments demands new techniques, especially if endowments are unbounded, see for example \cite{Cv}, \cite{kram04}, \cite{Khasanov2013} and \cite{SLZ}. The references \cite{KarZit03}, \cite{Zit05}, \cite{Mos} and \cite{Yu15} also study this problem when the intermediate consumption is considered.

In the presence of market frictions, the utility maximization problem relies heavily on the definition of \emph{working portfolio processes}. The conventional analysis based on semimartingale properties and stochastic integrals will not work in the general setting with transaction costs. New approaches are therefore required. In a multi-asset model, the self-financing admissible portfolios are defined carefully using the convex solvency cones and strictly consistent price systems (SCPS). The super-hedging theorems are developed under different market assumptions, see among \cite{MAFI:MAFI004}, \cite{Kab}, \cite{MAFI:MAFI180} and \cite{Campi06}. In a simple setting with one bond and one risky asset, the admissible portfolios are defined by requiring that the liquidation value process stays above some constant lower bound, see \cite{Sch2} and \cite{Sch33}. Recently in \cite{Sch33}, an easy-to-apply version of the super-hedging theorem is established under the assumption that the stock price process admits a CPS for arbitrarily small transaction costs.

As an important add-on to the existing literature, this paper aims to study the utility maximization problem under transaction costs together with unbounded random endowments. We note that the optimal investment problem with random endowments has been studied firstly by \cite{campi1}. In order to apply the super-hedging theorem in \cite{Campi06}, however, \cite{campi1} still works with admissible portfolios and their random endowments are assumed to satisfy $\mathcal{E}_T\in\mathbb{L}^{\infty}$ in order to guarantee the existence of the optimal solution. When the boundedness assumption is relaxed, the definition of admissible portfolios becomes no longer suitable and needs to be modified as the constant lower bound will turn out to be  an unnatural constraint. In the frictionless market, a definition of acceptable portfolios is introduced by \cite{DS97} and \cite{kram04} in which some maximal elements in the set of wealth processes can serve as the stochastic thresholds. However, the same choice of the maximal element from the admissible portfolio processes in the context of transaction costs can not be applied as lower bounds, see \cite{jacka2007} for some counterexamples.

Recently, in the Kabanov's multi-asset framework where transaction costs are modeled by matrix-valued processes, a new definition of acceptable portfolios is proposed in \cite{Yu} using convex solvency cones and SCPS. One of the main contributions in this work is to give a definition of acceptable portfolios in a simple setting as in \cite{Sch2} and \cite{Sch33}. In this paper, the self-financing portfolio is called acceptable if liquidation value processes are bounded below by some processes related to all CPS $(\mathbb{Q}, \tilde{S})$. The main idea behind our definition is to choose some maximal elements as stochastic thresholds from the set of wealth processes without transaction costs when the CPS $\tilde{S}$ is taken as the underlying asset. Comparing with the Definition $2.3$ in \cite{Yu} (see also Definition $2.7$ in \cite{Campi06}), it is worth noting that the condition of self-financing portfolio in our framework is more explicit with nice financial interpretations. However, we need to pay the price that it is generally more difficult to prove the super-hedging result as we need to verify a certain limit of a sequence of self-financing processes is still self-financing. With the assistance of convex solvency cones, this convergence result is taken as granted in \cite{Campi06} and \cite{Yu}. Meanwhile, unlike Definition $2.3$ in \cite{Yu} where the stochastic lower bounds are mandated on each portfolio process, we focus on each liquidation value process instead. Some new mathematical challenges arise due to the lack of the supermartingale property as in Lemma $2.8$ of \cite{Campi06} to prove the closedness property of the set of acceptable portfolios. In particular, the backward implication in Proposition $\ref{auxlem}$ of this paper is crucial for us to obtain the closedness result which is not needed in \cite{Campi06} and \cite{Yu}. Given the assumption that the stock price process admits CPS for all small transaction costs, we eventually are able to verify Proposition $\ref{auxlem}$ and thereby establish a super-hedging result using acceptable portfolios. The existence and uniqueness of the optimal solution are consequent on the duality theorem which is built upon the super-hedging result and convex analysis.

This paper also contributes to the application of the duality theorem to the existence of a shadow price process. Roughly speaking, a process $\hat{S}$ is called a shadow price if it evolves inside the bid-ask spread and the optimal frictionless trading in $\hat{S}$ leads to the same utility value function as in the original market under transaction costs and two optimal portfolio processes coincide. As stated in \cite{Chris22}, a candidate shadow price process is defined by $\hat{S}\triangleq \frac{Y^{1,\ast}}{Y^{0,\ast}}$ where $(Y^{0,\ast}, Y^{1,\ast})$ is the minimizer in the duality theorem. If the stock price process $S$ is \cadlag, $\hat{S}$ may not be a semimartingale as it may fail to be \cadlag. To overcome this difficulty, \cite{Chris22} considers a shadow price process $\hat{\mathbf{S}}=(\hat{S}^p, \hat{S})$ defined in a general \textit{sandwiched sense} such that $\hat{S}^p=\frac{Y^{1,\ast,p}}{Y^{0,\ast,p}}$ and $\hat{S}=\frac{Y^{1,\ast}}{Y^{0,\ast}}$ where $((Y^{0,\ast,p}, Y^{1,\ast,p}), (Y^{0,\ast}, Y^{1,\ast}))$ is a \textit{sandwiched strong supermartingale deflator}, see Definition $\ref{sandsuper}$ and Definition $\ref{sanddefla}$. Despite that the shadow price process fails to be \cadlag, the stochastic integrals are still well defined using predictable processes of finite variation as integrands. In \cite{Chris22}, the modified self-financing and admissible portfolio processes can therefore be defined and the verification of the shadow price process can be completed. With unbounded random endowments, the definition of sandwiched shadow price process given \cite{Chris22} can be extended in our setting using the modified acceptable portfolios. To the best of our knowledge, the study of a shadow price process in observing random endowments is new to the literature and we hope to add some interesting perspectives.

The rest of the paper is organized as follows: Section $\ref{section2}$ introduces the market model with transaction costs and the definition of acceptable portfolio processes. The utility maximization problem with unbounded random endowments is formulated in Section $\ref{section3}$. The dual space and the corresponding dual optimization problem are introduced afterwards. The main result of the duality theory is presented at the end. Section $\ref{section4}$ provides some sufficient conditions and establishes the existence of a sandwiched shadow price process consisting of a predictable and an optional strong supermartingales. The existence of a shadow price process in the usual sense is also discussed. Section $\ref{section5}$ contains the proofs of main theorems and all auxiliary results.

\section{Market Model}\label{section2}

We consider the market model which consists with one riskless bond and one risky asset. The riskless bond $B$ is assumed to be constant $1$ which amounts to serve as the num\'{e}raire. The stock price is modeled by a strictly positive and locally bounded adapted \cadlag process $(S_t)_{0\leq t\leq T}$ on some filtered probability space $(\Omega, \mathcal{F}, (\mathcal{F}_t)_{0\leq t\leq T},\mathbb{P})$ satisfying the usual assumptions of right continuity and completeness. The time horizon is given by $T>0$. Moreover, we assume that $\mathcal{F}_0$ is trivial, $\mathcal{F}_T=\mathcal{F}_{T-}$ and $S_T=S_{T-}$. Trading the risky asset incurs transaction costs, that is to say, we can buy the stock at the price $S$ but can only sell it at the price $(1-\lambda)S$. Here, $S$ denotes the ask price, $(1-\lambda)S$ denotes the bid price and $[(1-\lambda)S, S]$ is called the bid-ask spread.

\begin{definition}
For a given price process $S=(S_t)_{0\leq t\leq T}$ and transaction costs $0<\lambda<1$, a $\lambda$-consistent price system ($\lambda$-CPS) is a pair $(\mathbb{Q}, \tilde{S})$ such that $\mathbb{Q}$ is a probability measure equivalent to $\mathbb{P}$, $\tilde{S}=(\tilde{S}_t)_{0\leq t\leq T}$ takes its values in the bid-ask spread $[(1-\lambda)S, S]=([(1-\lambda)S_t, S_t])_{0\leq t\leq T}$ and $\tilde{S}$ is a $\mathbb{Q}$-local martingale.
\end{definition}

Denote $\mathcal{S}(\lambda, S)$ ( short as $\mathcal{S}$) as the set of all $\tilde{S}$ such that $(\mathbb{Q},\tilde{S})$ is a CPS with transaction costs $\lambda$.  For each $\tilde{S}\in\mathcal{S}$, also denote set $\mathcal{M}(\tilde{S})$ as the set of all probability measures $\mathbb{Q}$ such that $(\mathbb{Q}, \tilde{S})$ is a $\lambda$-CPS. Define the set $\mathcal{M}(\lambda, S)$ (short as $\mathcal{M}$) by $\mathcal{M}\triangleq \bigcup_{\tilde{S}\in\mathcal{S}}\mathcal{M}(\tilde{S})$. Notice that each $\tilde{S}$ is a semimartingale under the physical probability measure $\mathbb{P}$. Given the initial wealth $a>0$, denote $\mathcal{X}(\tilde{S}, a)$ as the set of all nonnegative wealth processes in the $\tilde{S}$-market, $\tilde{S}\in\mathcal{S}$. That is
\begin{align*}
\mathcal{X}(\tilde{S}, a)\triangleq \Big\{X\geq 0: &X_t=a+(H\cdot\tilde{S})_t,\ \ \text{where $H$ is predictable}\\
&\text{and $\tilde{S}$-integrable,\ $t\in[0,T]$} \Big\}.
\end{align*}
A wealth process in $\mathcal{X}(\tilde{S}, a)$ is called \textit{maximal}, denoted by $X^{\max, \tilde{S}}$, if its terminal value $X_T^{\max,\tilde{S}}$ can not be dominated by any other processes in $\mathcal{X}(\tilde{S}, a)$.

\begin{assumption}\label{Assum}
For each $0<\lambda'<1$, the price process $S$ admits $\lambda'$-CPS.
\end{assumption}

The trading strategy $\phi=(\phi^0, \phi^1)_{0\leq t\leq T}$ represents the holdings in units of the riskless and the risky asset, respectively, after rebalancing the portfolios at time $t$. $(\phi^0,\phi^1)$ is called \textit{self-financing with transaction costs $\lambda$} (see \cite{Sch2} and \cite{Chris22}) if
\begin{itemize}
\item[(i)] $\phi=(\phi^0,\phi^1)_{0\leq t\leq T}$ is a pair of predictable processes of finite variation.

\item[(ii)] For any process $\phi$ of finite variation, $\phi=x+\phi^{\uparrow}-\phi^{\downarrow}$ represents its Jordan-Hahn decomposition into two non-decreasing processes $\phi^{\uparrow}$ and $\phi^{\downarrow}$ both null at zero. $(\phi^0, \phi^1)$ satisfies the condition
\begin{equation}\label{self}
\int_{s}^{t}d\phi_u^0 \leq -\int_s^t S_ud\phi_u^{1,\uparrow}+\int_s^t(1-\lambda)S_ud\phi_u^{1,\downarrow}
\end{equation}
a.s. for all $0\leq s<t\leq T$, where
\begin{equation}
\int_s^t S_ud\phi_u^{1,\uparrow}\triangleq \int_s^t S_u d\phi_u^{1,\uparrow, c}+\sum_{s<u\leq t}S_{u-}\triangle \phi_{u}^{1,\uparrow}+\sum_{s\leq u<t}S_u\triangle_{+}\phi_u^{1,\uparrow},\nonumber
\end{equation}
and
\begin{align*}
\int_s^t(1-\lambda)S_ud\phi_u^{1,\downarrow}\triangleq &\int_s^t(1-\lambda)S_ud\phi_u^{1,\downarrow,c}\\
&+\sum_{s<u\leq t}(1-\lambda)S_{u-}\triangle \phi_u^{1,\downarrow}+\sum_{s\leq u<t}(1-\lambda)S_u\triangle_+\phi_u^{1,\downarrow}\nonumber
\end{align*}
can be defined as Riemann-Stieltjes integrals since $S$ is \cadlag. Here we define $\triangle \phi_t\triangleq \phi_t-\phi_{t-}$ and $\triangle_+ \phi_t\triangleq \phi_{t+}-\phi_t$.
\end{itemize}

It is worth noting that as $S$ is \cadlag, we need to take care of both left and right jumps of the portfolio process $\phi$. In general, three values $\phi_{\tau-}$, $\phi_{\tau}$ and $\phi_{\tau+}$ may be different. If the stopping time $\tau$ is totally inaccessible, the predictability of $\phi$ implies that $\triangle \phi_{\tau}=0$ almost surely. But if the stopping time $\tau$ is predictable, it may happen that both $\triangle \phi_{\tau}\neq 0$ and $\triangle_+\phi_{\tau}\neq 0$.

Given the initial position $(\phi^0_0,\phi^1_0)=(x,0)$ in the bond and risky asset separately, where $x\in\mathbb{R}$, we define the \textit{liquidation value} at time $t$ by
\begin{equation}
V(\phi)_t\triangleq \phi_t^0+(\phi_t^1)^+(1-\lambda)S_t-(\phi_t^1)^-S_t.\nonumber
\end{equation}

The conventional definition of working portfolios in the existing literature assumes constant thresholds for the liquidation value processes, see \cite{Sch2}:
\begin{definition}\label{admissiblephi}
For an $\mathbb{R}_+$-valued adapted \cadlag process $S=(S_t)_{0\leq t\leq T}$ with transaction costs $0<\lambda<1$, a self-financing trading strategy $\phi$ is called admissible if there exists a constant $a\geq 0$ and for every $[0,T]$-valued stopping time $\tau$,
\begin{equation}
V(\phi)_{\tau}=\phi_{\tau}^0+(\phi_{\tau}^1)^+(1-\lambda)S_{\tau}-(\phi_{\tau}^1)^-S_{\tau}\geq -a, \ \ \text{a.s.}\nonumber
\end{equation}
\end{definition}

From now on, the market is enlarged by allowing trading $N$ European contingent claims at time $t=0$ with final payoff $\mathcal{E}_T=(\mathcal{E}_T^i)_{1\leq i\leq N}$. We denote $q=(q^i)_{1\leq i\leq N}$ as static holdings in contingent claims $\mathcal{E}_T$.
By allowing $q$ to take negative values, without loss of generality, we can only consider the case $\mathcal{E}_T^i\geq 0$ for all $1\leq i\leq N$. Each $\mathcal{E}_T^i$ may be unbounded, but it is assumed that $\sum_{i=1}^{N}\mathcal{E}_T^i$ is integrable uniformly with respect to the set $\mathcal{M}$ in the following sense:
\begin{assumption}\label{assE}
\begin{equation}\label{unic}
\lim_{m\rightarrow\infty}\sup_{\mathbb{Q}\in\mathcal{M}}\mathbb{E}^{\mathbb{Q}}\bigg[ \bigg(\sum_{i=1}^{N}\mathcal{E}_T^i\bigg)\mathbf{1}_{\{\sum_{i=1}^{N}\mathcal{E}_T^i>m\}}\bigg]=0.
\end{equation}
\end{assumption}

Clearly, $(\ref{unic})$ implies the finite super-hedging price $\sup_{\mathbb{Q}\in\mathcal{M}}\mathbb{E}^{\mathbb{Q}}[\sum_{i=1}^{N}\mathcal{E}^i_T]<\infty$ of the payoff $\sum_{i=1}^{N}\mathcal{E}^i_T$. Indeed, similar to the proof of de la Vall\'{e}e-Poussin theorem of uniformly integrability, it is straightforward to verify the following equivalent condition for Assumption $\ref{assE}$.
\begin{lemma}\label{lemmmmm}
Assumption $\ref{assE}$ holds if and only if there exists a Borel test function $\phi(x)$ with $\lim_{x\rightarrow\infty}\frac{\phi(x)}{x}=\infty$ such that
\begin{equation}\label{testineqphi}
\sup_{\mathbb{Q}\in\mathcal{M}}\mathbb{E}^{\mathbb{Q}}[\phi(X)]<\infty,
\end{equation}
where we define $X=\sum_{i=1}^{N}\mathcal{E}_T^i$. If it exists, the function $\phi(x)$ can be chosen in the class of non-decreasing convex functions. In particular, if for some $p>1$, the $p$-th moment of the random endowment $\mathcal{E}_T$ is super-hedgeable under all $\lambda$-CPS, i.e.,
\begin{equation}\label{momentp}
\sup_{\mathbb{Q}\in\mathcal{M}}\mathbb{E}^{\mathbb{Q}}[X^p]<\infty,
\end{equation}
Assumption $\ref{assE}$ is satisfied.
\end{lemma}

For bounded random endowments that $\mathcal{E}_T\in\mathbb{L}^{\infty}$, Assumption $\ref{assE}$ holds trivially. Lemma $\ref{lemmmmm}$ states that it is sufficient to require $q\cdot\mathcal{E}_T\in\mathbb{L}^p(\mathbb{Q})$ for some $p>1$ and all $\mathbb{Q}\in\mathcal{M}$. Assumption $\ref{assE}$ is a mathematical condition that we need later for the proof of the super-hedging result.

The following result holds (see the proof of Lemma $2.1$ in \cite{Yu}).
\begin{lemma}\label{boundE}
Under Assumption $\ref{assE}$, there exists a constant $a>0$ such that for each $\tilde{S}\in\mathcal{S}$, there exits a maximal element $X^{\max,\tilde{S}}\in\mathcal{X}(\tilde{S}, a)$ and $\sum_{i=1}^{N}\mathcal{E}_T^i\leq X_T^{\max,\tilde{S}}$.
\end{lemma}

\begin{assumption}\label{assNo}
For any $q\in\mathbb{R}^N$ such that $q\neq 0$, the random variable $q\cdot\mathcal{E}_T$ is not replicable in the market under $\lambda$-CPS.
\end{assumption}

To deal with unbounded random endowments, the above definition of admissible portfolios is not appropriate. The constant lower bound needs to be relaxed as the stochastic threshold. Following the idea of \cite{Yu}, we shall propose the modified working portfolios as below.
\begin{definition}\label{accptpot}
For an $\mathbb{R}_+$-valued adapted c\`{a}dl\`{a}g process $S=(S_t)_{0\leq t\leq T}$ with transaction costs $0<\lambda<1$, a self-financing trading strategy $\phi$ is called \textbf{acceptable} if there exists a constant $a\geq 0$ and for each $\tilde{S}\in\mathcal{S}$, there exists a maximal element $X^{\max,\tilde{S}}\in\mathcal{X}(\tilde{S}, a)$ such that for every $[0,T]$-valued stopping time $\tau$,
\begin{equation}
V(\phi)_{\tau}=\phi_{\tau}^0+(\phi_{\tau}^1)^+(1-\lambda)S_{\tau}-(\phi_{\tau}^1)^-S_{\tau} \geq -X^{\max,\tilde{S}}_{\tau},\ \ \text{a.s.}\nonumber
\end{equation}
\end{definition}
\begin{remark}
Each admissible portfolio process is acceptable as any given constant $a>0$ is a maximal element in $\mathcal{X}(\tilde{S},a)$. Indeed, for each $\tilde{S}\in\mathcal{S}$, there exists $\mathbb{Q}\sim\mathbb{P}$ such that $\tilde{S}$ is a $\mathbb{Q}$-local martingale. It follows that each $\tilde{S}$ is a semimartingale and satisfies the No Free Lunch with Vanishing Risk condition, see \cite{Sch94} for details. Therefore a contradiction arises if there exists a maximal element in $\mathcal{X}(\tilde{S},a)$ which dominates the constant $a$.
\end{remark}

Denote by $\mathcal{A}_x(\lambda, S)$ (short as $\mathcal{A}_x$) the set of all pairs $(\phi^0, \phi^1)\in\mathbb{L}^0(\mathbb{R}^2)$ of acceptable portfolios with transaction costs $\lambda$ starting at $\phi_{0}=(\phi_{0}^0,\phi_{0}^{1})=(x,0)$. We call $\mathcal{U}_x(\lambda, S)$ (short as $\mathcal{U}_x$) the set of all terminal values of the pair $(\phi^0,\phi^1)\in\mathcal{A}_x$, i.e., $\mathcal{U}_x=\{(\phi^0_T,\phi^1_T): (\phi^0,\phi^1)\in\mathcal{A}_x\}$. Let us also denote $\mathcal{V}_x(\lambda, S)$ (short as $\mathcal{V}_x$) the set of all terminal values of these liquidation value processes such that the position in the stock is liquidated at time $T$, i.e., $\mathcal{V}_x=\{V_T: V_T=\phi_T^0,\ \phi_T^1=0,\ \ (\phi^0,\phi^1)\in\mathcal{A}_x\}$.

Contrary to the admissible portfolios, the definition of acceptable portfolio in our setting seems more difficult to check because it involves all $\tilde{S}\in\mathcal{S}$ and all $0\leq t\leq T$. The following result, however, asserts that it is sufficient to check the terminal time $T$.

\begin{proposition}\label{auxlem}
Fix the c\`{a}dl\`{a}g, adapted process $S$ and transaction costs $0<\lambda<1$ as above and let Assumption $\ref{Assum}$ hold. Fix $\hat{a}>0$ and for each $\tilde{S}\in\mathcal{S}$, pick and fix one $\hat{X}^{\max,\tilde{S}}\in\mathcal{X}(\tilde{S},\hat{a})$. For any $(\phi^0,\phi^1)\in\mathcal{A}_x$ and for each $\tilde{S}\in\mathcal{S}$, if we have
\begin{equation}\label{later}
V(\phi^0,\phi^1)_T=\phi_T^0+(\phi_T^1)^+(1-\lambda)S_T-(\phi_T^1)^-S_T\geq -\hat{X}_T^{\max,\tilde{S}},
\end{equation}
then for every $[0,T]$-valued stopping time $\tau$, we also have
\begin{equation}\label{bef}
V(\phi^0,\phi^1)_{\tau}=\phi_{\tau}^0+(\phi_{\tau}^1)^+(1-\lambda)S_{\tau}-(\phi_{\tau}^{1})^-S_{\tau}\geq -\hat{X}_{\tau}^{\max,\tilde{S}}.
\end{equation}
\end{proposition}

Proposition $\ref{auxlem}$ provides a convenient way to check the definition of acceptable portfolios. If there exists a random variable $B$ which satisfies $\sup_{\mathbb{Q}\in\mathcal{M}}\mathbb{E}^{\mathbb{Q}}[B]<\infty$ and $V(\phi^0, \phi^1)_T\geq -B$, Proposition $\ref{auxlem}$ together with Lemma $\ref{boundE}$ imply that the self-financing portfolio $(\phi^0,\phi^1)$ is acceptable. More importantly, the backward implication in Proposition $\ref{auxlem}$ can replace the super-martingale property in the later proof of the super-hedging theorem.

\section{Utility Maximization with Unbounded Random Endowments}\label{section3}

We first introduce the primal set of acceptable portfolio processes with the initial wealth $x\in\mathbb{R}$ whose terminal liquidation value dominates the payoff $-q\cdot\mathcal{E}_T$ by
\begin{equation}\label{primeH}
\mathcal{H}(x,q)\triangleq \{ V_T: V_T+q\cdot\mathcal{E}_T \geq 0,\ \ V\in\mathcal{V}_x\}.
\end{equation}
The effective domain is defined by $\mathcal{K}\triangleq \text{int}\left\{ (x,q)\in\mathbb{R}^{1+N}:\mathcal{H}(x,q)\neq \emptyset\right\}$.

The agent's preference is represented by a utility function $U:(0,\infty)\rightarrow\mathbb{R}$, which is assumed to be strictly increasing, strictly concave and continuously differentiable. It is assumed that the utility function satisfies the Inada conditions $U'(0)\triangleq\lim_{x\rightarrow 0}U'(x)=\infty$ and $U'(\infty)\triangleq \lim_{x\rightarrow\infty}U'(x)=0$. Moreover, we make the assumption on the asymptotic elasticity of the utility function
\begin{equation}\label{asym}
AE(U)\triangleq \underset{x\rightarrow\infty}{\lim\sup}\frac{xU'(x)}{U(x)}<1.
\end{equation}
The convex conjugate of $U(x)$ is defined by $\tilde{U}(y)\triangleq \sup_{x>0}\Big(U(x)-xy\Big)$, $y>0$.

Given $(x,q)\in\mathcal{K}$, the agent is to maximize the expected utility defined on the terminal wealth consisting of the terminal liquidation value and the final payoff from the contingent claims. The \textbf{primal utility optimization problem} is defined by
\begin{equation}\label{primeu}
u(x,q)\triangleq \sup_{V_T\in\mathcal{H}(x,q)}\mathbb{E}[U(V_T+q\cdot\mathcal{E}_T)],\ \ \ (x,q)\in\mathcal{K}.
\end{equation}

Let $\mathcal{C}(x,q)$ be the solid hull of the primal set $\mathcal{H}(x,q)$
\begin{equation}\label{primsp}
\mathcal{C}(x,q)\triangleq \{g\in\mathbb{L}_+^0: g\leq V_T+q\cdot\mathcal{E}_T,\ \ V_T\in\mathcal{H}(x,q)\},\ \ (x,q)\in\mathcal{K}.
\end{equation}
The monotonicity of $U(x)$ implies that $u(x,q)=\sup_{g\in\mathcal{C}(x,q)}\mathbb{E}[U(g)]$, $(x,q)\in\mathcal{K}$.

Following \cite{kram04}, we consider the \textit{relative interior} of the polar cone of $-\mathcal{K}$ defined by
\begin{equation}
\mathcal{L}\triangleq \text{ri}\{(y,r)\in\mathbb{R}^{1+N}: xy+q\cdot r\geq 0\ \ \text{for all}\ (x,q)\in\mathcal{K}\}.\nonumber
\end{equation}

Denote $\mathcal{B}$ as the set of density processes of $\lambda$-CPS that
\begin{align*}
\mathcal{B}\triangleq  \bigg\{(Z^0, Z^1)\geq 0: &Z_t^0=\mathbb{E}\Big[\frac{d\mathbb{Q}}{d\mathbb{P}}\Big|\mathcal{F}_t\Big],\ \ \text{and}\ \ Z_t^1=\tilde{S}_tZ_t^0,\\
&\text{where}\ \mathbb{Q}\in\mathcal{M}(\tilde{S}),\ \text{for each}\ \tilde{S}\in\mathcal{S} \bigg\}.
\end{align*}

In general, the set $\mathcal{B}$ lacks the closedness property and a proper enlargement is needed for it to serve as a dual set of $\mathcal{C}(x,q)$.
\begin{definition}
Starting with a strictly positive initial position $(\phi^0_0,\phi^1_0)=(x,0)$ where $ x>0$, the admissible portfolio $(\phi^0,\phi^1)$ is called $0$-admissible if for every $[0,T]$-valued stopping time $\tau$, the liquidation value process satisfies $V(\phi)_{\tau}\geq 0$, a.s. Given $x>0$, we shall denote the set of all $0$-admissible portfolio by $\mathcal{A}_x^{\text{adm}}$ and the set of all terminal values of the $0$-admissible portfolio by $\mathcal{U}_x^{\text{adm}}$, i.e.,
\begin{equation}\label{0admU}
\mathcal{U}_x^{\text{adm}}=\{(\phi_T^0,\phi^1_T)\in\mathbb{L}^0(\mathbb{R}^2): (\phi^0,\phi^1)\in\mathcal{A}_x^{\text{adm}}\},\ \ x>0.
\end{equation}
We also denote $\mathcal{V}_x^{\text{adm}}$ as the set of the terminal value of all $0$-admissible liquidation value processes with initial position $(x,0)$ such that the position in the stock is liquidated at $t=T$, i.e.,
\begin{equation}\label{0admV}
\mathcal{V}_x^{\text{adm}}=\{V_T\in\mathbb{L}^0_+(\mathbb{R}):  \exists (\phi_T^0,\phi_T^1)\in\mathcal{U}_x^{\text{adm}}\ \text{such that}\ \phi_T^0=V_T,\ \phi_T^1=0\}.
\end{equation}
\end{definition}

As it is assumed that $S=(S_t)_{0\leq t\leq T}$ is c\`{a}dl\`{a}g, all self-financing portfolio processes $(\phi^0_t, \phi_t^1)_{0\leq t\leq T}$ need to be predictable of finite variation and can have both left and right jumps in order to obtain that $\mathcal{U}_x^{\text{adm}}$ is closed under convergence in probability, see \cite{Campi06} and \cite{Sch33} for details. To retain supermartingale properties, a new limit is required to replace Fatou's limit; see \cite{Chris} and \cite{Chris22}. The convergence in probability at all finite stopping times and the concept of \textit{optional strong supermartingales} seem to be tailor-made to analyze problems with transaction costs. The following definition in \cite{Chris} plays an important role in the definition of the dual set.

\begin{definition}
An optional process $X=(X_t)_{0\leq t\leq T}$ is called an optional strong supermartingale if, for all stopping times $0\leq \sigma\leq \tau\leq T$, we have
\begin{equation}
\mathbb{E}[X_{\tau}|\mathcal{F}_{\sigma}]\leq X_{\sigma},\nonumber
\end{equation}
where we impose that $X_{\tau}$ is integrable for any $[0,T]$-valued stopping time $\tau$.
\end{definition}

We shall enlarge the dual set using the optional strong supermartingales. For $y>0$,
\begin{align}\label{deflator}
\mathcal{Z}(y)\triangleq \bigg\{&(Y^0, Y^1)\ \text{are nonnegative optional strong supermartingales}:\nonumber\\
&Y_0^0=y, \frac{Y^1}{Y^0}\in[(1-\lambda)S, S], \phi^0Y^0+\phi^1Y^1\ \text{is a non-negative}\nonumber\\
&\text{optional strong supermartingale},\ \forall(\phi^0,\phi^1)\in\mathcal{A}_1^{\text{adm}}\bigg\},
\end{align}
and
\begin{equation}\label{dualY}
\mathcal{Y}(y)\triangleq \{Y_T\in\mathbb{L}_+^0(\mathbb{R}): \exists (Y^0, Y^1)\in\mathcal{Z}(y)\ \text{with}\ Y_T=Y_T^0\},\ \ y>0.
\end{equation}
Due to Proposition $1.6$ in \cite{Sch2}, we have that $y\mathcal{B}\subset\mathcal{Z}(y)$.

Given $(y,r)\in\mathcal{L}$, we are interested in the subset
\begin{equation}\label{dualsp}
\mathcal{Y}(y,r)\triangleq \{Y_T\in\mathcal{Y}(y): \mathbb{E}[Y_T(V_T+q\cdot\mathcal{E}_T)]\leq xy+q\cdot r,\ \ V_T\in\mathcal{H}(x,q),\ \ (x,q)\in\mathcal{K}\},
\end{equation}
which is the proposed dual set to work on as random endowments can be hidden by its definition.

Let the abstract set $\mathcal{D}(y,r)$ be the solid hull of $\mathcal{Y}(y,r)$,
\begin{equation}
\mathcal{D}(y,r)=\{h\in\mathbb{L}_+^0(\mathbb{R}^2): h\leq Y_T,\ \ Y_T\in\mathcal{Y}(y,r)\},\ \ (y,r)\in\mathcal{L}.\nonumber
\end{equation}

We can then define the corresponding \textbf{dual optimization problem} to problem $(\ref{primeu})$ by
\begin{equation}\label{dualv}
v(y,r)\triangleq  \inf_{Y_T\in\mathcal{Y}(y,r)}\mathbb{E}[\tilde{U}(Y_T)]=\inf_{h\in\mathcal{D}(y,r)}\mathbb{E}[\tilde{U}(h)],\ \ (y,r)\in\mathcal{L}.
\end{equation}

The following duality theorem provides the existence and uniqueness of the optimal solution to the utility maximization problem $(\ref{primeu})$.
\begin{theorem}\label{mainthm}
Let Assumptions $\ref{Assum}$, $\ref{assE}$ and $\ref{assNo}$ and condition $(\ref{asym})$ hold. Furthermore, we assume that $u(x,q)<\infty$ for some $(x,q)\in\mathcal{K}$. Then we have
\begin{itemize}
\item[(i)] The function $u$ is finitely valued on $\mathcal{K}$ and the function $v$ is finitely valued on $\mathcal{L}$. The value functions $u$ and $v$ are conjugate
\begin{equation}
\begin{split}
u(x,q)=\inf_{(y,r)\in\mathcal{L}}\Big( v(y,r)+xy+q\cdot r\Big),\ \ \ (x,q)\in\mathcal{K},\\
v(y,r)=\sup_{(x,q)\in\mathcal{K}}\Big(u(x,q)-xy-q\cdot r \Big),\ \ \ (y,r)\in\mathcal{L}.\nonumber
\end{split}
\end{equation}

\item[(ii)] The optimal solution $Y_T^{\ast}(y,r)$ to $(\ref{dualv})$ exists and is unique for all $(y,r)\in\mathcal{L}$.

\item[(iii)] The optimal solution $V_T^{\ast}(x,q)$ to $(\ref{primeu})$ exists and is unique for all $(x,q)\in\mathcal{K}$.

\item[(iv)] There are $(\phi^{0,\ast}, \phi^{1,\ast})\in\mathcal{A}_x$ and $(Y^{0,\ast}, Y^{1,\ast})\in\mathcal{Z}(y)$ such that
\begin{equation}
V(\phi^{0,\ast}, \phi^{1,\ast})_T=V_T^{\ast}(x,q),\ \ \text{and}\ \ \ Y_T^{0,\ast}=Y_T^{\ast}(y,r).\nonumber
\end{equation}

\item[(v)] The super-differential of $u$ maps $\mathcal{K}$ into $\mathcal{L}$, i.e.,
\begin{equation}
\partial u(x,q)\subset \mathcal{L},\ \ (x,q)\in\mathcal{K}.\nonumber
\end{equation}
\item[(vi)] If $(y,r)\in\partial u(x,q)$, the optimal solutions are related by
\begin{equation}\label{ok1}
\begin{split}
Y^{\ast}_T(y,r)&=U'(V^{\ast}_T(x,q)+q\cdot \mathcal{E}_T),\\
\mathbb{E}[Y^{\ast}_T(y,r)(V_T^{\ast}(x,q)+q\cdot\mathcal{E}_T)]&=xy+q\cdot r.
\end{split}
\end{equation}
\end{itemize}
\end{theorem}

\begin{remark}
Denote $\mathcal{P}(x,q;U)$ the set of all marginal utility-based prices at $(x,q)\in\mathcal{K}$
\begin{equation}\label{marprice}
\mathcal{P}(x,q;U)\triangleq\{p\in\mathbb{R}: u(x-q'p, q+q')\leq u(x,q)\ \text{for all}\ q'\in\mathbb{R}\}.
\end{equation}
The definition asserts that the agent's holdings $q$ in $\mathcal{E}_T$ is optimal in the model where the contingent claims can be traded at the marginal utility-based price $p$ at time zero. Equivalently, see \cite{kram04} and \cite{Hugo}, we have
\begin{equation}\label{marprice2}
\mathcal{P}(x,q;U)= \Big\{\frac{r}{y}: (y,r)\in\partial u(x,q)\Big\}.
\end{equation}
The duality theorem above can serve as the first step to perform the sensitivity analysis and first order expansion of the marginal utility-based prices similar to \cite{Sirbu2006} and \cite{Kramkov20071606}.
\end{remark}

\section{Connections to the Shadow Prices}\label{section4}
In this section, we apply the duality theorem to study the existence of the shadow price process in the frictionless market with random endowments. To simplify the notation, we shall take $N=1$ and hence $q\in\mathbb{R}$.

First, we introduce the concept of a shadow price in the usual sense. To this end, we need some preparations of definitions. For a fixed $\lambda$-CPS $(\mathbb{Q}, \hat{S})$, i.e., $\hat{S}\in\mathcal{S}$ and the positive initial wealth $x>0$, we define the set of self-financing and $0$-admissible trading strategies in the market without transaction costs by
\begin{equation}
\begin{split}
\mathcal{A}_x^{\text{adm}}(\hat{S})\triangleq \bigg\{(\phi^0, \phi^1): \ \ &x+\int_0^t\phi_u^1d\hat{S}_u\geq 0,\ \forall t\in[0,T],\ \ \text{$\phi^1$ is predictable}\\
&\text{and $\hat{S}$-integrable},\ \phi_t^0=x+\int_0^t\phi_u^1d\hat{S}_u-\phi_t^1\hat{S}_t\bigg\}.\nonumber
\end{split}
\end{equation}

The set of wealth processes with $0$-admissible strategies in the $\hat{S}$-market is define by, for $x>0$,
\begin{equation}
\mathcal{X}(\hat{S},x)\triangleq \bigg\{X: X_t=x+\int_0^t\phi_u^1d\hat{S}_u\geq 0,\ \ \forall t\in[0,T],\ \ (\phi^0,\phi^1)\in\mathcal{A}_x^{\text{adm}}(\hat{S}) \bigg\}.\nonumber
\end{equation}
We denote $X^{\max}$ the maximal element in the set $\mathcal{X}(\hat{S}, x)$ for some $x>0$ and $\mathcal{X}(\hat{S})\triangleq \bigcup_{x>0}\mathcal{X}(\hat{S}, x)$.

\begin{definition}\label{notransactions}
For a fixed $\hat{S}\in\mathcal{S}$, the self-financing portfolio is called acceptable in the frictionless $\hat{S}$-market if the wealth process $X$ admits a representation $X=X'-X^{\max}$, where $X'$ is a wealth process under some $0$-admissible portfolios and $X^{\max}$ is a maximal element in $\mathcal{X}(\hat{S})$. That is to say, the set of all acceptable portfolios can be written as
\begin{equation}\label{shadowaccpt}
\begin{split}
\mathcal{A}_x(\hat{S})\triangleq \bigg\{(\phi^0, \phi^1):\ \ &x+\int_0^t\phi_u^1d\hat{S}_u=X'_t-X_t^{\max},\ \phi^1\ \text{is predictable}\\
&\text{and $\hat{S}$-integrable}, \text{where}\ X', X^{\max}\in\mathcal{X}(\hat{S})\\
&\text{and}\ \phi_t^0=x+\int_0^t\phi_u^1d\hat{S}_u-\phi_t^1\hat{S}_t,\ \forall t\in[0,T] \bigg\}.\nonumber
\end{split}
\end{equation}
\end{definition}

The set of all terminal wealth processes in the $\hat{S}$-market is denoted by
\begin{equation}
\mathcal{V}_x(\hat{S})\triangleq \bigg\{X_T\in\mathbb{L}^0(\mathbb{R}): X_T=x+\int_0^T\phi_u^1d\hat{S}_u=\phi_T^0+\phi_T^1\hat{S}_T,\ \ (\phi^0,\phi^1)\in\mathcal{A}_x(\hat{S}) \bigg\},\nonumber
\end{equation}
and the set of terminal wealth values under acceptable portfolios dominating the payoff $-q\mathcal{E}_T$ is defined by
\begin{equation}
\mathcal{H}(x,q;\hat{S})\triangleq \{X_T: X_T+q\mathcal{E}_T\geq 0,\ X_T\in\mathcal{V}_x(\hat{S}) \}.\nonumber
\end{equation}
The corresponding effective domain is given by
\begin{equation}
\mathcal{K}(\hat{S})\triangleq \text{int} \{(x,q)\in\mathbb{R}^{2}: \mathcal{H}(x,q;\hat{S})\neq \emptyset \}.\nonumber
\end{equation}

\begin{definition}\label{shadowprice}
A process $\hat{S}\in\mathcal{S}$, i.e., $(\mathbb{Q}, \hat{S})$ is a $\lambda$-CPS for some $\mathbb{Q}\in\mathcal{M}(\hat{S})$, is called a \textit{shadow price process}, if the optimal solution $(\hat{\phi}^0, \hat{\phi}^1)$ with the terminal wealth $X(\hat{\phi}^0,\hat{\phi}^1)_T\in\mathcal{H}(x,q;\hat{S})$ to the frictionless utility maximization problem
\begin{equation}\label{frictionless}
u(x,q;\hat{S})\triangleq \sup_{X_T\in\mathcal{H}(x,q;\hat{S})}\mathbb{E}[U(X_T+q\mathcal{E}_T)],
\end{equation}
exists for $(x,q)\in\mathcal{K}\cap\mathcal{K}(\hat{S})$ and coincides with the optimal solution $\phi^{\ast}=(\phi^{0,\ast}, \phi^{1, \ast})$ to the problem $(\ref{primeu})$ under transaction costs $\lambda$. In particular, we have $u(x,q)=u(x,q;\hat{S})$.
\end{definition}

\begin{remark}
Our definition of the classic shadow price process $\hat{S}$ is more restrictive than \cite{Chris22} and \cite{Chris33} and a shadow price process $\hat{S}$ satisfies the NFLVR condition by its definition. The acceptable portfolio differs from the admissible portfolio and the existence of equivalent local martingale measures for $\hat{S}$ are required to build the duality relationship in the shadow price market without transaction costs. Therefore, unlike \cite{Chris33}, even when $S$ is continuous, the existence of consistent local martingale system $(Z^0,Z^1)$ (see \cite{BayYu} for its definition and the equivalent characterization) is no longer a sufficient condition for the existence of a shadow price process with random endowments.
\end{remark}

\begin{remark}
Comparing Definition $\ref{accptpot}$ and Definition $\ref{notransactions}$, it is easy to see that $\mathcal{A}_x\subseteq\mathcal{A}_x(\hat{S})$ since we require $\hat{S}\in\mathcal{S}$. Therefore, it follows that $\mathcal{H}(x,q)\neq\emptyset$ implies that $\mathcal{H}(x,q;\hat{S})\neq\emptyset$ and hence $\mathcal{K}\subset\mathcal{K}(\hat{S})$. In Definition $\ref{shadowprice}$, it is then enough to require $(x,q)\in\mathcal{K}$ for the well-posedness of both $u(x,q)$ and $u(x,q;\hat{S})$.
\end{remark}

If a shadow price $\hat{S}$ exists, an optimal strategy $(\hat{\phi})=(\hat{\phi}^0, \hat{\phi}^1)$ for the utility maximization problem $(\ref{frictionless})$ in the frictionless market can be realized in the market with transaction costs. In particular, we aim to show that the optimal strategy $(\phi^{0,\ast}, \phi^{1,\ast})$ to the problem $(\ref{primeu})$ under transaction costs only trades if $\hat{S}$ is at the bid or ask price, i.e.,
\begin{equation}
\{d\phi^{1,\ast}>0\}\subseteq\{\hat{S}=S\},\ \ \text{and}\ \ \ \{d\phi^{1,\ast}<0\}\subseteq\{\hat{S}=(1-\lambda)S\}\nonumber
\end{equation}
in the sense that
\begin{align}\label{stratno}
\{d\phi^{1,\ast,c}>0\} &\subseteq\{\hat{S}=S\},\ \ \ \ \{d\phi^{1,\ast,c}<0\}\subseteq\{ \hat{S}=(1-\lambda)S\},\nonumber\\
\{\triangle \phi^{1,\ast}>0 \}&\subseteq \{\hat{S}_{-}=S_{-}\},\ \ \ \ \{ \triangle \phi^{1,\ast}<0\}\subseteq \{\hat{S}_{-}=(1-\lambda)S_{-}\},\nonumber\\
\{\triangle_{+} \phi^{1,\ast}>0 \}&\subseteq \{\hat{S}=S\},\ \ \ \ \{ \triangle_{+} \phi^{1,\ast}<0\}\subseteq \{\hat{S}=(1-\lambda)S\}.
\end{align}

Define the dual set for the shadow price $\hat{S}$ by
\begin{align*}
\mathcal{Y}(y; \hat{S})\triangleq \bigg\{Y\geq 0:& Y_0=y\ \text{and}\ Y_t(\phi_t^0+\phi_t^1\hat{S}_t)=Y_t\bigg(1+\int_0^t\phi^1_ud\hat{S}_u\bigg),\\
&\text{is a \cadlag supermartingale for all $(\phi^0, \phi^1)\in\mathcal{A}_1^{\text{adm}}(\hat{S})$} \bigg\}.
\end{align*}

The \text{relative interior} of the polar cone of $-\mathcal{K}(\hat{S})$ is denoted by
\begin{equation}
\mathcal{L}(\hat{S})\triangleq \text{ri}\{(y,r)\in\mathbb{R}^{2}: xy+qr\geq 0\ \text{for all}\ (x,q)\in\mathcal{K}(\hat{S})\}.\nonumber
\end{equation}

Define $\mathcal{Y}(y,r;\hat{S})$ as the subset of $\mathcal{Y}(y;\hat{S})$ by
\begin{align*}
\mathcal{Y}(y,r;\hat{S})\triangleq \{Y_T\in\mathcal{Y}(y;\hat{S}):&\ \mathbb{E}[Y_T(X_T+q\mathcal{E}_T)]\leq xy+qr,\\
&\ X_T\in\mathcal{H}(x,q;\hat{S}),\ (x,q)\in\mathcal{K}(\hat{S})\}.\nonumber
\end{align*}

The dual optimization problem to $(\ref{frictionless})$ is then formulated as
\begin{equation}\label{shadowdualv}
v(y,r;\hat{S})=\inf_{Y_T\in\mathcal{Y}(y,r;\hat{S})}\mathbb{E}[\tilde{U}(Y_T)].
\end{equation}

Example $4.1$ of \cite{Chris22} shows that if $S$ is \cadlag, the dual optimizer $(Y^{0,\ast}, Y^{1,\ast})$ to problem $(\ref{dualv})$ and the candidate of the shadow price process $\hat{S}\triangleq \frac{Y^{1,\ast}}{Y^{0,\ast}}$ may not be \cadlag and therefore may not be semimartingales. The existence of a shadow price may fail in general. However, the stochastic integral $\int_0^t\hat{\phi}^1_ud\hat{S}_u$ can still be defined as long as $\hat{\phi}^1$ is a predictable process of finite variation and $\hat{S}$ is \ladlag (see \cite{Chris} and \cite{Chris22}) and
\begin{equation}\label{wrongsha}
\int_0^t\hat{\phi}^1_ud\hat{S}_u=\int_0^t\hat{\phi}_u^{1,c}d\hat{S}_u+\sum_{0<u\leq t}\triangle\hat{\phi}_u^{1}(\hat{S}_t-\hat{S}_{u-})+\sum_{0\leq u<t}\triangle_+\hat{\phi}_u^1(\hat{S}_t-\hat{S}_u),\ \ 0\leq t\leq T.
\end{equation}
The integral above can still be interpreted as gains from the trading of the self-financing portfolio $(\hat{\phi}^1_t)_{0\leq t\leq T}$ without transaction costs under the price process $\hat{S}=(\hat{S}_t)_{0\leq t\leq T}$, although $\hat{S}$ is not a semimartingale. Therefore, the natural question is that whether or not can we choose the quotient $\hat{S}=\frac{Y^{1,\ast}}{Y^{0,\ast}}$ as the underlying asset and define the wealth process in this general shadow price market by the stochastic integral $(\ref{wrongsha})$? Unfortunately, the answer is negative in general. Example $4.2$ in \cite{Chris22} points out that we may not be able to verify properties $(\ref{stratno})$ using the wealth process defined by $(\ref{wrongsha})$. In particular, it is difficult to guarantee that
\begin{equation}
\{\triangle \phi^{1,\ast}>0 \}\subseteq \{\hat{S}_{-}=S_{-}\},\ \ \ \ \{ \triangle \phi^{1,\ast}<0\}\subseteq \{\hat{S}_{-}=(1-\lambda)S_{-}\},\nonumber
\end{equation}
where $\phi^{1,\ast}$ is the optimal portfolio process in Theorem $\ref{mainthm}$. As a consequence, we are not able to verify that $(\phi^{0,\ast}, \phi^{1,\ast})$ is the optimal solution in the shadow price market driven by $\hat{S}=\frac{Y^{1,\ast}}{Y^{0,\ast}}$. It requires us to modify either the definition of $\hat{S}$ or the wealth process given by $(\ref{wrongsha})$.

To examine the shadow price process in a correct generalized form, Example $4.2$ in \cite{Chris22} shows the importance of the following concepts.

\begin{definition}
A predictable process $X=(X_t)_{0\leq t\leq T}$ is called a predictable strong supermartingale if, for all predictable stopping times $0\leq \sigma\leq \tau\leq T$, we have
\begin{equation}
\mathbb{E}[X_{\tau}|\mathcal{F}_{\sigma}]\leq X_{\sigma},\nonumber
\end{equation}
where we impose that $X_{\tau}$ is integrable for any $[0,T]$-valued predictable stopping time $\tau$.
\end{definition}

\begin{definition}\label{sandsuper}
A sandwiched strong supermartingale is a pair $\mathbf{X}=(X^p,X)$ such that $X^p$ (resp. $X$) is a predictable (resp. optional) strong supermartingale and such that
\begin{equation}\label{sand}
X_{\tau-}\geq X_{\tau}^p\geq \mathbb{E}[X_{\tau}|\mathcal{F}_{\tau-}],
\end{equation}
for all predictable stopping times $\tau$.
\end{definition}

For a sandwiched strong supermartingale $\mathbf{X}=(X^p, X)$ and a predictable process $\phi$ of finite variation, as in \cite{Chris}, the stochastic integral is defined in \textit{a sandwiched sense} by
\begin{equation}
\int_0^t\phi_ud\mathbf{X}_u\triangleq \int_0^t\phi_u^cdX_u+\sum_{0< u\leq t}\triangle\phi_u(X_t-X_u^p)+\sum_{0\leq u< t}\triangle_+\phi_u(X_t-X_u),\ \ 0\leq t\leq T. \nonumber
\end{equation}

\begin{definition}\label{sanddefla}
We call $\mathbf{Y}=(Y^p, Y)=((Y^{0,p}, Y^{1,p}), (Y^0, Y^1))$ a sandwiched strong supermartingale deflator if $Y=(Y^0, Y^1)\in\mathcal{Z}(y)$ (see $(\ref{deflator})$) and $(Y^{0,p}, Y^0)$ and $(Y^{1,p}, Y^1)$ are sandwiched strong supermartingales and the process $\hat{S}^p$ stays inside the bid-ask spread,
\begin{equation}
\hat{S}^p_t=\frac{Y^{1,p}}{Y^{0,p}}\in\ [(1-\lambda)S_{t-}, S_{t-}],\ \ t\in[0,T]. \nonumber
\end{equation}
\end{definition}

Following the proof of Lemma $A.1$ of \cite{Chris22}, by passing to the forward convex combinations if necessary, we have the following convergence results.
\begin{lemma}\label{seplemma}
Fix $(x,q)\in\mathcal{K}$. For any $(y,r)\in\partial u(x,q)$, there exists a minimizing sequence $Z^n(y,r)=(Z_t^{0,n}(y,r), Z_t^{1,n}(y,r))_{0\leq t\leq T}$ in $\mathcal{B}(1)$ to the dual problem $(\ref{dualv})$, i.e.,
\begin{equation}
\mathbb{E}[\tilde{U}(yZ_T^{0,n}(y,r))]\searrow v(y,r),\ \ \text{as}\ n\rightarrow\infty,\nonumber
\end{equation}
and a sandwiched strong supermartingale deflator $\mathbf{Y}^{\ast}(y,r)=(Y^{\ast, p}(y,r), Y^{\ast}(y,r))$ such that
\begin{equation}\label{con111}
(yZ^{0,n}_{\tau-}(y,r), yZ^{1,n}_{\tau-}(y,r))\xrightarrow{\mathbb{P}}(Y^{0,\ast,p}_{\tau}(y,r), Y^{1,\ast,p}_{\tau}(y,r)),
\end{equation}
and
\begin{equation}\label{con222}
(yZ^{0,n}_{\tau}(y,r), yZ^{1,n}_{\tau}(y,r))\xrightarrow{\mathbb{P}}(Y^{0,\ast}_{\tau}(y,r), Y^{1,\ast}_{\tau}(y,r)),
\end{equation}
as $n\rightarrow\infty$ for all $[0,T]$-valued stopping time $\tau$, where $Y^{0,\ast}(y,r)$ is the dual optimizer to $(\ref{dualv})$.
\end{lemma}

To ensure the existence of a sandwiched shadow price process in the next Theorem $\ref{theorem1}$, the following assumption is needed for some technical reasons. 

\begin{assumption}\label{importantnew}
Fix $(x,q)\in\mathcal{K}$. Assume that there exists some $(y,r)\in\partial u(x,q)$ such that the minimizing sequence $Z^n(y,r)=(Z_t^{0,n}(y,r), Z_t^{1,n}(y,r))_{0\leq t\leq T}$ in $\mathcal{B}(1)$ to the dual problem $(\ref{dualv})$ satisfies
\begin{equation}\label{marginalprice}
\underset{n\rightarrow\infty}{\lim\inf}\ \mathbb{E}[Z_T^{0,n}(y, r)\mathcal{E}_T]=\frac{r}{y}.
\end{equation}
\end{assumption}

Denote $\mathcal{P}$ the set of all arbitrage-free prices. For any $Z\in\mathcal{B}(1)$, we have $\mathbb{E}[Z_T^{0}\mathcal{E}_T]\in\mathcal{P}$ and $\mathcal{P}(x,q;U)\subset\mathcal{P}$, where $\mathcal{P}(x,q;U)$ is the set of all marginal utility-based prices, see $(\ref{marprice})$ and $(\ref{marprice2})$ for its definition. The condition $(\ref{marginalprice})$ requires the existence of a marginal utility based price $\frac{r}{y}\in\mathcal{P}(x,q;U)$ which can be achieved by a minimizing sequence $Z^{0,n}(y,r)$. In other words, the limit infimum of the arbitrage free prices under the minimizing sequence $Z^{0,n}$ equals $\frac{r}{y}$, i.e., $\underset{n\rightarrow\infty}{\lim\inf}\ \mathbb{E}[Z_T^{0,n}(y, r)\mathcal{E}_T]=\frac{r}{y}$. Here, we reveal a sufficient condition for the existence of a sandwiched shadow price process related to the property of some marginal utility-based prices. The following two examples provide some concrete market models satisfying the condition $(\ref{marginalprice})$ for separate cases when $q>0$ and $q<0$.

\begin{example}
We assume that $\mathcal{E}_T\leq (1-\lambda)S_T$. Suppose that for $(x,q)\in\mathcal{K}$ and $q>0$, there exists some marginal utility based price $(y,r)\in\partial u(x,q)$ of $\mathcal{E}_T$ which satisfies $(1-\lambda)S_0\leq \frac{r}{y} \leq S_0$. As the initial value of $Y_0^{1,\ast}$ is flexible. Without loss of generality, we can consider the initial value of $Y_0^{1,\ast}$ as $Y_0^{1,\ast}(y,r)=r$, which satisfies the bid-ask spread constraint
\begin{equation}
(1-\lambda)S_0\leq \frac{Y_0^{1,\ast} (y,r)}{Y_0^{0,\ast} (y,r)} \leq S_0
\end{equation}

Consider the minimizing sequence $(Z^{0,n}(y,r), Z^{1,n}(y,r))\in\mathcal{B}(1)$, we have $\tilde{S}^n\triangleq \frac{Z^{1,n}(y,r)}{Z^{0,n}(y,r)}\in[(1-\lambda)S, S]$. Therefore, it is easy to see that
\begin{align*}
\mathbb{E}[Z_T^{0,n}(y, r)\mathcal{E}_T]&\leq a\mathbb{E}[Z_T^{0,n}(y,r)(1-\lambda)S_T]\leq a\mathbb{E}[Z_T^{0,n}(y,r)\tilde{S}_T^n]\\
&=\mathbb{E}[Z_T^{1,n}(y,r)]\leq Z_0^{1,n}(y,r).
\end{align*}
As $yZ_0^{1,n}(y,r)$ converges to $Y_0^{1,\ast}(y,r)=r$, it follows that
\begin{equation}\label{exampleinq}
\underset{n\rightarrow\infty}{\lim\inf}\mathbb{E}[Z_T^{0,n}(y, r)\mathcal{E}_T]\leq \underset{n\rightarrow\infty}{\lim\inf}Z_0^{1,n}(y,r)\leq \frac{r}{y}.
\end{equation}

On the other hand, for the same pair $(y,r)\in\partial u(x,q)$, we have that
\begin{align}\label{exampleinq2}
xy+qr=&\mathbb{E}[Y_T^{0,\ast}(y,r)(V_T^{\ast}(x,q)+q\mathcal{E}_T)]\leq\underset{n\rightarrow\infty}{\lim\inf}\mathbb{E}[yZ_T^{0,n}(y,r)(V_T^{\ast}(x,q)+q\mathcal{E}_T)]\nonumber\\
\leq & xy+qy\underset{n\rightarrow\infty}{\lim\inf}\mathbb{E}[Z_T^{0,n}(y,r)\mathcal{E}_T].
\end{align}
As $q>0$, it follows that $\underset{n\rightarrow\infty}{\lim\inf}\mathbb{E}[Z_T^{0,n}(y,r)\mathcal{E}_T]\geq \frac{r}{y}$. The last inequality and $(\ref{exampleinq})$ yield $(\ref{marginalprice})$.
\end{example}

\begin{example}
Assume that $\mathcal{E}_T\geq S_T$. Suppose that for the choice of $(x,q)\in\mathcal{K}$ and $q<0$, there exists some marginal utility based price $(y,r)\in\partial u(x,q)$ of $\mathcal{E}_T$ which satisfies
\begin{equation}\label{exampleinq3}
\mathbb{E}[Y_T^{1,\ast}(y,r)]\geq r^{\ast},
\end{equation}
where $r^{\ast}$ is defined as the smallest value of $r$ such that $(y,r)\in\partial u(x,q)$.

Because for any $(y,r)$ and $(\bar{y},\bar{r})\in\partial u(x,q)$, we always have $Y_T^{0,\ast}(y,r)=Y_T^{0,\ast}(\bar{y},\bar{r})$. We shall pick the pair $(y^{\ast}, r^{\ast})\in\partial u(x,q)$. For the minimizing sequence $(Z^{0,n}(y^{\ast},r^{\ast}), Z^{1,n}(y^{\ast},r^{\ast}))\in\mathcal{B}(1)$, Fatou's lemma together with $(\ref{exampleinq3})$ and the fact that $\frac{Y^{1,\ast}(y^{\ast}, r^{\ast})}{Y^{0,\ast}(y^{\ast}, r^{\ast})}\in [(1-\lambda)S, S]$ imply that
\begin{align*}
\underset{n\rightarrow\infty}{\lim\inf}\mathbb{E}[Z_T^{0,n}(y^{\ast},r^{\ast})\mathcal{E}_T]\geq &\mathbb{E}\Big[\frac{1}{y^{\ast}}Y_T^{0,\ast}(y^{\ast}, r^{\ast})\mathcal{E}_T\Big]\geq \mathbb{E}\Big[\frac{1}{y^{\ast}}Y_T^{0,\ast}(y^{\ast}, r^{\ast})S_T\Big]\\
\geq &\mathbb{E}\bigg[\frac{1}{y^{\ast}}Y_T^{0,\ast}(y^{\ast}, r^{\ast}) \frac{Y_T^{1,\ast}(y^{\ast}, r^{\ast})}{Y^{0,\ast}_T(y^{\ast}, r^{\ast})}\bigg]\\
=&\mathbb{E}\Big[\frac{1}{y^{\ast}}Y_T^{1,\ast}(y^{\ast}, r^{\ast})\Big]\geq \frac{r^{\ast}}{y^{\ast}}.
\end{align*}
For the same pair $(y^{\ast},r^{\ast})\in\partial u(x,q)$, following $(\ref{exampleinq2})$ and the fact $q<0$, we will have
\begin{equation}
\underset{n\rightarrow\infty}{\lim\inf}\mathbb{E}[Z_T^{0,n}(y^{\ast},r^{\ast})\mathcal{E}_T]\leq\frac{r^{\ast}}{y^{\ast}},\nonumber
\end{equation}
which verifies $(\ref{marginalprice})$ with the choice of $r=r^{\ast}$ and $y=y^{\ast}$.
\end{example}

\begin{remark}
Assumption $\ref{importantnew}$ is in general not straightforward to verify and may hold valid under certain conditions on $\mathcal{E}_T$ and the choices of $(x,q)$. The major difficulty in the proof of the next Theorem $\ref{theorem1}$ is that for the fixed choice of $(y,r)$ in the dual problem, we can not compare the value of $\mathbb{E}[Y_T^{0,\ast}(y,r)q\mathcal{E}_T]$ or the approximating sequence $\mathbb{E}[yZ_T^{0,n}(y,r)q\mathcal{E}_T]$ with the value $r$ as the product $Y^{0,\ast}(y,r)V(\phi^{0,\ast}(x,q),\phi^{1,\ast}(x,q))$ may not be a martingale in general. Existence of a sandwiched shadow price without Assumption $\ref{importantnew}$ is a challenging but interesting problem, which will be left as a future research project. 
\end{remark}

\begin{remark}
It should be possible to extend the framework of \cite{SLZ} for unbounded random endowments and then establish the duality theorem using the pair of bounded finitely additive measures which admits the Yosida-Hewitt decomposition $\mathbb{Q}^i=\mathbb{Q}^{i,r}+\mathbb{Q}^{i,s}$, $i=0,1$, and $\mathbb{Q}^{i, r}$ is a countably additive measure. It then might be easier to check Assumption $\ref{importantnew}$ in this extended framework. However, one should also be aware of some new challenges in this framework: 1. One needs to verify that $\hat{S}=\frac{Y^{1,\ast}}{Y^{0,\ast}}$ induced by the (regular part of) cluster point of a minimizing net of finitely additive measures will still stay in the bid-ask spread for all time; 2. On the other hand, another new difficulty is to verify the dual optimizers in two models are the same, which will be critical to guarantee that the optimal portfolios in two markets can coincide. This becomes nontrivial because one may expect that regular parts of two dual optimizers are the same while it might be difficult to conclude that values $\langle \mathbb{Q}^{0,\ast}, \mathcal{E}_T\rangle$ in two dual models also coincide as the singular parts in the primal market and the shadow market may differ. In conclusion, although the extension of the framework in \cite{SLZ} with unbounded random endowments might open the door to check the existence of shadow price without Assumption $\ref{importantnew}$, many further technical efforts are still required. 
\end{remark}

We are now ready to present the next main result which provides the existence of a sandwiched supermartingale deflator related to the dual minimizer of the problem $(\ref{dualv})$, and hence the candidate sandwiched shadow price process is well-defined.

\begin{theorem}\label{theorem1}
Fix $(x,q)\in\mathcal{K}$. Under all assumptions of Theorem $\ref{mainthm}$ and under Assumption $\ref{importantnew}$, there exists at least a pair of $(y,r)\in\partial u(x,q)$ and for the optimizer $\phi^{\ast}(x,q)=(\phi^{0,\ast}(x,q), \phi^{1,\ast}(x,q))$ to the primal utility maximization problem $(\ref{primeu})$, we have
\begin{equation}\label{equiphi}
Y^{0,\ast}(y,r)\phi^{0,\ast}(x,q)+Y^{1,\ast}(y,r)\phi^{1,\ast}(x,q)=Y^{0,\ast}(y,r)(x+\phi^{1,\ast}(x,q)\cdot\hat{\mathbf{S}}),
\end{equation}
where
\begin{equation}
\hat{\mathbf{S}}=(\hat{S}^p, \hat{S})=\Big(\frac{Y^{1,\ast,p}(y, r)}{Y^{0,\ast,p}(y, r)}, \frac{Y^{1,\ast}(y, r)}{Y^{0,\ast}(y, r)}\Big)\nonumber
\end{equation}
and
\begin{align}\label{selfphixq}
(\phi^{1,\ast}(x,q)\cdot\hat{\mathbf{S}})_t\triangleq &\int_0^t\phi^{1,\ast,c}_u(x,q)d\hat{S}_u\nonumber\\
&+\sum_{0\leq u<t}\triangle \phi^{1,\ast}_u(x,q)(\hat{S}_t-\hat{S}^p_u)+\sum_{0<u\leq t}\triangle_+\phi_u^{1,\ast}(x,q)(\hat{S}_t-\hat{S}_u).
\end{align}
It follows that
\begin{equation}\label{behavphi}
\begin{split}
\{d\phi^{1,\ast,c}(x,q)>0\}&\subseteq \{\hat{S}=S\},\ \ \ \ \{d\phi^{1,\ast,c}(x,q)<0\}\subseteq\{\hat{S}=(1-\lambda)S\},\\
\{\triangle\phi^{1,\ast}(x,q)>0\}&\subset\{\hat{S}^p=S_{-}\},\ \ \ \ \{\triangle\phi^{1,\ast}(x,q)<0\}\subseteq \{\hat{S}^p=(1-\lambda)S_{-}\},\\
\{\triangle_+\phi^{1,\ast}(x,q)>0\}&\subset\{\hat{S}=S\},\ \ \ \ \ \{\triangle_+\phi^{1,\ast}(x,q)<0\}\subseteq \{\hat{S}=(1-\lambda)S\}.
\end{split}
\end{equation}
\end{theorem}

For any sandwiched supermartingale deflator $\mathbf{Y}=(Y^p,Y)$ with the associated price process $\hat{\mathbf{S}}=(\hat{S}^p,\hat{S})=\Big(\frac{Y^{1,p}}{Y^{0,p}}, \frac{Y^{1}}{Y^0}\Big)$, and any acceptable trading strategy $\phi\in\mathcal{A}_x$, it is easy to verify that the liquidation value $V(\phi^0, \phi^1)$ satisfies
\begin{align*}
V(\phi^0,\phi^1)_t&=\phi_t^0+(\phi_t^1)^+(1-\lambda)S_t-(\phi_t^1)^{-}S_t\\
&\leq x+\int_0^t\phi_u^{1,c}d\hat{S}_u+\sum_{0< u\leq t}\triangle \phi_u^1(\hat{S}_t-\hat{S}^p_u)+\sum_{0\leq u< t}\triangle_+\phi_u^1(\hat{S}_t-\hat{S}_u)\\
&=x+(\phi^1\cdot\hat{\mathbf{S}})_t.
\end{align*}

Thanks to $(\ref{equiphi})$ and $(\ref{behavphi})$, we are able to verify that the optimal strategy $(\phi^{0,\ast}, \phi^{1,\ast})$ only trades when the sandwiched shadow price process $\hat{\mathbf{S}}=(\hat{S}^p, \hat{S})$ assumes the least favorable position in the bid-ask spread.

In order to verify the existence of the sandwiched shadow price process, it is important to give a new definition of the acceptable portfolios for the underlying price process $\hat{\mathbf{S}}$. Clearly, the definition of $\mathcal{A}_x(\hat{S})$ in $(\ref{shadowaccpt})$ is too wide in general because we can only work with integrand processes of finite variation. The equality $(\ref{equiphi})$ gives us a hint of the definition of self-financing portfolios for the sandwiched shadow prices. Let us also recall that the important property behind the concept of a sandwiched shadow price $\hat{\mathbf{S}}$ is that any self-financing and acceptable portfolio trading with $\hat{\mathbf{S}}$ can not do better than the optimizer $(\phi^{0,\ast},\phi^{1,\ast})$ given in Theorem $\ref{mainthm}$ for the price process  $S$ with transaction costs $\lambda$. Moreover, the strategy $(\phi^{0,\ast}, \phi^{1,\ast})$ trading in $\hat{\mathbf{S}}$ without transaction costs brings the same expected utility value as the case of trading in $S$ under transaction costs $\lambda$. Similar to the definition of admissible portfolios in \cite{Chris22}, we can now give the following modified definition of self-financing and acceptable portfolios for the sandwiched shadow price process such that it is comparable with respect to the definition of acceptable portfolios for $S$ with transaction costs $\lambda$.

\begin{definition}\label{sandSaccpt}
The portfolio process $(\phi^0_t, \phi^1_t)_{0\leq t\leq T}$ is called \textbf{acceptable} for the sandwiched shadow price process $\hat{\mathbf{S}}$ if
\begin{itemize}
\item[(i)] $(\phi^0, \phi^1)$ is predictable process of finite variation.
\item[(ii)] $(\phi^0, \phi^1)$ is self-financing for $\hat{\mathbf{S}}$ without transaction costs in the sense that
\begin{equation}
\phi_t^0=x+\int_0^t\phi_u^1d\hat{\mathbf{S}}_u-\phi_t^1\hat{S}_t,\ \ 0\leq t\leq T.\nonumber
\end{equation}
\item[(iii)] Define the auxiliary liquidation value process by
\begin{equation}
V(\phi^0,\phi^1)_t\triangleq \phi_t^0+(\phi_t^1)^+(1-\lambda)S_t-(\phi_t^1)^-S_t.\nonumber
\end{equation}
There exists a constant $a>0$ such that for each $\tilde{S}\in\mathcal{S}$, there exists a maximal element $X^{\max,\tilde{S}}\in\mathcal{X}(\tilde{S}, a)$ and $V(\phi^0, \phi^1)_{\tau}\geq -X^{\max,\tilde{S}}_{\tau}$
for all $[0,T]$-valued stopping time $\tau$.
\end{itemize}
\end{definition}
Denote $\mathcal{A}_x(\hat{\mathbf{S}})$ the set of all acceptable portfolio processes for the sandwiched shadow price process $\hat{\mathbf{S}}$ starting with initial position $(\phi^0_0, \phi^1_0)=(x,0)$. Also, denote $\mathcal{V}_x(\hat{\mathbf{S}})$ the set of terminal value of all wealth processes generated by acceptable portfolios
\begin{equation}
\mathcal{V}_x(\hat{\mathbf{S}})\triangleq \bigg\{X_T: X_T=\phi_T^0+\phi_T^1\hat{S}_T=x+\int_0^T\phi_u^1d\hat{\mathbf{S}}_u,\ \ (\phi^0, \phi^1)\in\mathcal{A}_x(\hat{\mathbf{S}})\bigg\}.\nonumber
\end{equation}

Similar to the case of shadow price process in the usual sense, given the same random endowment $\mathcal{E}_T$ and initial static position $q\in\mathbb{R}$, let us consider the primal set
\begin{equation}
\mathcal{H}(x,q;\hat{\mathbf{S}})\triangleq \{X_T:\  X_T+q\mathcal{E}_T\geq 0,\ X_T\in\mathcal{V}_x(\hat{\mathbf{S}})\},\ \ (x,q)\in\mathcal{K}(\hat{\mathbf{S}})\nonumber
\end{equation}
where we define $\mathcal{K}(\hat{\mathbf{S}})\triangleq \text{int}\{(x,q)\in\mathbb{R}^{2}: \mathcal{H}(x,q;\hat{\mathbf{S}})\neq\emptyset\}$.

The next theorem concerns the existence of a sandwiched shadow price process.

\begin{theorem}\label{theosandsha}
Fix $(x,q)\in\mathcal{K}$. Under all assumptions in Theorem $\ref{theorem1}$, let $(y,r)\in\partial u(x,q)$ satisfy Assumption $\ref{importantnew}$ and let $(Y^{0,\ast}(y,r), Y^{1,\ast}(y,r))$ be the dual optimizer of problem $(\ref{dualv})$. Consider any $X(\phi^0,\phi^1)_T\in\mathcal{H}(x,q;\hat{\mathbf{S}})$ for the sandwiched shadow price process defined by $\hat{\mathbf{S}}=(\hat{S}^p,\hat{S})=\Big(\frac{Y^{1,\ast,p}(y,r)}{Y^{0,\ast,p}(y,r)}, \frac{Y^{1,\ast}(y,r)}{Y^{0,\ast}(y,r)}\Big)$. We have
\begin{align*}
&\mathbb{E}\Big[U(X(\phi^0,\phi^1)_T+q\mathcal{E}_T)\Big]\leq \mathbb{E}\bigg[U(x+\int_0^T\phi^{1,\ast}_ud\hat{\mathbf{S}}_u+q\mathcal{E}_T)\bigg]\\
=&\mathbb{E}\Big[U(\phi_T^{0,\ast}+\phi_T^{1,\ast}\hat{S}_T +q\mathcal{E}_T)\Big]=\mathbb{E}\Big[U(V(\phi^{0,\ast}, \phi^{1,\ast})_T+q\mathcal{E}_T)\Big],
\end{align*}
where $(\phi^{0,\ast}(x,q), \phi^{1,\ast}(x,q))$ is the optimal solution to the primal utility maximization problem $(\ref{primeu})$.
\end{theorem}

Comparing with Proposition $3.7$ in \cite{Chris22} and Theorem $3.1$ in \cite{Chris33}, the existence of a shadow price process under random endowments becomes more delicate and can fail in general. In our framework, it is even not enough to require that the dual optimizer $(Y^{0,\ast}(y,r), Y^{1,\ast}(y,r))$ satisfies the condition that $Y^{0,\ast}(y,r)$ is a martingale and $Y^{1,\ast}(y,r)$ is a local martingale. Actually, first, we need to require that the classic shadow price process admits NFLVR condition so that the duality theory can be obtained in the shadow price market. Second, in order to check that the dual optimizer $Y^{0,\ast}(y,r)$ is in the dual space $\mathcal{Y}(y,r;\hat{S})$ of the shadow price market and to compare utility value functions in two corresponding markets, we have to make the assumption that $Y^{0,\ast}(y,r)\in y\mathcal{M}(\frac{r}{y})\subset \mathcal{Y}(y,r)$ where we define $\mathcal{M}(p)=\{\mathbb{Q}\in\mathcal{M}: \mathbb{E}^{\mathbb{Q}}[\mathcal{E}_T]=p\}$, $p\in\mathcal{P}(x,q;U)$ and $\mathcal{P}(x,q;U)$ is the set of all marginal utility-based prices. Therefore, it is assumed that there exists some $(y,r)\in\partial u(x,q)$ such that the arbitrage-free price of $\mathcal{E}_T$ under the measure $\frac{d\mathbb{Q^{\ast}}}{d\mathbb{P}}=\frac{1}{y}Y_T^{0,\ast}(y,r)$ equals the chosen marginal utility-based price, i.e.,  $\mathbb{E}^{\mathbb{Q}^{\ast}}[\mathcal{E}_T]=\frac{r}{y}$.

The next theorem summarizes the existence of a shadow price in the usual sense under some sufficient conditions discussed above.
\begin{theorem}\label{classical1}
Fix $(x,q)\in\mathcal{K}$ and $q>0$ and consider some $(y,r)\in\partial u(x,q)\subset\mathcal{L}$. If the dual minimizer $(Y^{0,\ast}(y,r), Y^{1,\ast}(y,r))$ to the problem $(\ref{dualv})$ satisfies that $(Y^{0,\ast}(y,r), Y^{1,\ast}(y,r))\in y\mathcal{B}$ and $Y^{0,\ast}_T(y,r)\in y\mathcal{M}\Big(\frac{r}{y}\Big)$. The process $\hat{S}(y,r)$ defined by $\hat{S}(y,r)\triangleq \frac{Y^{1,\ast}(y,r)}{Y^{0,\ast}(y,r)}$ is a classical shadow price process given in Definition $\ref{shadowprice}$ to the utility maximization problem $(\ref{primeu})$ with the price process $S$ and the transaction costs $\lambda$.
\end{theorem}

\section{Proofs of Main Results}\label{section5}
This section contains proofs of all main theorems and auxiliary results in the previous sections.

\subsection{Proof of Proposition $\ref{auxlem}$}
\begin{proof}
The proof of Theorem $1.7$ in \cite{Sch2} can be modified to our setting using acceptable portfolios. Assume that $(\ref{bef})$ does not hold for one fixed $\tilde{S}\in\mathcal{S}(\lambda, S)$, we may find $\frac{\lambda}{2}>\alpha>0$ and a stopping time $0\leq \tau\leq T$ such that either $\mathbb{P}(A_+)>0$ or $\mathbb{P}(A_{-})>0$, where we define
\begin{equation}
A_{+}\triangleq \Big\{\phi_{\tau}^{1}\geq 0,\ \phi_{\tau}^{0}+\phi_{\tau}^{1}\frac{1-\lambda}{1-\alpha}S_{\tau}<-\hat{X}_{\tau}^{\max,\tilde{S}} \Big\},\nonumber
\end{equation}
and
\begin{equation}
A_{-}\triangleq \Big\{\phi_{\tau}^{1}\leq 0, \phi_{\tau}^{0}+\phi_{\tau}^{1}(1-\alpha^2)S_{\tau}<-\hat{X}_{\tau}^{\max,\tilde{S}} \Big\}.\nonumber
\end{equation}

For the fixed $\tilde{S}\in\mathcal{S}(\lambda, S)$ and $\hat{X}^{\max,\tilde{S}}\in\mathcal{X}(\tilde{S},a)$, consider any $\mathbb{Q}\in\mathcal{M}(\tilde{S})$, $\hat{X}^{\max,\tilde{S}}$ is a supermartingale under $\mathbb{Q}$. Hence, $\mathbb{Q}(\hat{X}_{\tau}^{\max,\tilde{S}} \geq \hat{X}_T^{\max,\tilde{S}})>0$ holds for the previous stopping time $\tau$. Also, as $\mathbb{P}\sim\mathbb{Q}$, we deduce that $\mathbb{P}(\hat{X}_{\tau}^{\max,\tilde{S}} \geq \hat{X}_T^{\max,\tilde{S}})>0$. Let us define two auxiliary sets
\begin{equation}
B_+\triangleq \{\hat{X}_{\tau}^{\max,\tilde{S}} \geq \hat{X}_T^{\max,\tilde{S}}\}\cap A_+,\nonumber
\end{equation}
and
\begin{equation}
B_{-}\triangleq \{\hat{X}_{\tau}^{\max,\tilde{S}} \geq \hat{X}_T^{\max,\tilde{S}}\}\cap A_{-}.\nonumber
\end{equation}
Clearly it follows that $\mathbb{P}(B_+)>0$ or $\mathbb{P}(B_{-})>0$.

Choose $0<\lambda'<\alpha$ and consider a $\lambda'$-CPS with $\bar{S}$ taking values in the spread $[(1-\lambda')S, S]$ and $\mathbb{Q}\in\mathcal{M}(\bar{S}; \lambda')$, where we denote $\mathcal{M}(\bar{S}; \lambda')$ as the set of all $\mathbb{Q}$ such that $(\mathbb{Q}, \bar{S})$ is a $\lambda'$-CPS. It is easy to check that $(1-\alpha)\bar{S}$ and $\frac{1-\lambda}{1-\alpha}\bar{S}$ stays in the spread $[(1-\lambda)S, S]$, and it follows that for any $\mathbb{Q}\in\mathcal{M}(\bar{S}; \lambda')$, $(\mathbb{Q}, (1-\alpha)\bar{S})$ and $(\mathbb{Q}, \frac{1-\lambda}{1-\alpha}\bar{S})$ are both $\lambda$-CPS. Moreover, thanks to Proposition $1.6$ of \cite{Sch33}, we deduce that $\phi_t^0+\phi_t^1(1-\alpha)\bar{S}_t$ and $\phi_t^0+\phi_t^1\frac{1-\lambda}{1-\alpha}\bar{S}_t$, $0\leq t\leq T$ are both local optional strong $\mathbb{Q}$-supermartingales. Since $(\phi^0,\phi^1)$ is an acceptable portfolio, there exists a constant $a>0$ and for $\bar{S}\in\mathcal{S}(\lambda', S)\subset\mathcal{S}$, there exists a $X^{\max,\bar{S}}\in\mathcal{X}(\bar{S},a)$ as the lower bound. It follows that
\begin{equation}
0\leq V(\phi^0,\phi^1)_{\tau}+X_{\tau}^{\max,\bar{S}}\leq \phi_{\tau}^{0}+\phi_{\tau}^{1}(1-\alpha)\bar{S}_{\tau}+X_{\tau}^{\max,\bar{S}},\nonumber
\end{equation}
for any $[0,T]$-valued stopping time $\tau$. Therefore $\phi_{t}^{0}+\phi_{t}^{1}(1-\alpha)\bar{S}_{t}+X_t^{\max,\bar{S}}$ is an optional strong $\mathbb{Q}$-supermartingale for any $\mathbb{Q}\in\mathcal{M}(\bar{S}; \lambda')$. Consider the subset
\begin{equation}
\mathcal{M}'(\bar{S}; \lambda')\triangleq \{\mathbb{Q}\in\mathcal{M}(\bar{S};\lambda'): \text{$X^{\max,\bar{S}}$ is a UI martingale under $\mathbb{Q}$}\}.\nonumber
\end{equation}
For any fixed $\mathbb{Q}\in\mathcal{M}'(\bar{S}; \lambda')$, as $\bar{S}\geq (1-\alpha)S$, we obtain that
\begin{align*}
\mathbb{E}^{\mathbb{Q}}[V(\phi^0,\phi^1)_T| B_{-}]&\leq \mathbb{E}^{\mathbb{Q}}[\phi_T^0+\phi_T^1(1-\alpha)\bar{S}_T+X_T^{\max,\bar{S}}|B_{-}]-\mathbb{E}^{\mathbb{Q}}[X_T^{\max,\bar{S}}|B_{-}]\\
&\leq \mathbb{E}^{\mathbb{Q}}[\phi_{\tau}^0+\phi_{\tau}^1(1-\alpha)\bar{S}_{\tau}+X_{\tau}^{\max,\bar{S}}|B_{-}]-\mathbb{E}^{\mathbb{Q}}[X_{\tau}^{\max,\bar{S}}|B_{-}]\\
&=\mathbb{E}^{\mathbb{Q}}[\phi_{\tau}^0+\phi_{\tau}^1(1-\alpha)\bar{S}_{\tau}|B_{-}]\leq \mathbb{E}^{\mathbb{Q}}[\phi_{\tau}^{0}+\phi_{\tau}^{1}(1-\alpha)^2 S_{\tau}|B_{-}]\\
&< \mathbb{E}^{\mathbb{Q}}[-\hat{X}_{\tau}^{\max,\tilde{S}}|B_{-}]\leq \mathbb{E}^{\mathbb{Q}}[-\hat{X}_{T}^{\max,\tilde{S}}|B_{-}].
\end{align*}

Similarly, for the same lower bound $X^{\max,\bar{S}}$ chosen above, we have $\phi_t^0+\phi_t^1\frac{1-\lambda}{1-\alpha}\bar{S}_t+X_t^{\max,\bar{S}}$ is an optional strong $\mathbb{Q}$-supermartingale for any $\mathbb{Q}\in\mathcal{M}(\bar{S}; \lambda')$. Again, pick one $\mathbb{Q}\in\mathcal{M}'(\bar{S}; \lambda')$, we have
\begin{align*}
\mathbb{E}^{\mathbb{Q}}[ V(\phi^0,\phi^1)_T| B_+]&\leq \mathbb{E}^{\mathbb{Q}}\Big[\phi_T^0+\phi_T^1\frac{1-\lambda}{1-\alpha}\bar{S}_T+X_T^{\max,\bar{S}} \Big|B_+\Big]-\mathbb{E}^{\mathbb{Q}}[X_T^{\max,\bar{S}}|B_+]\\
&\leq \mathbb{E}^{\mathbb{Q}}\Big[\phi_{\tau}^0+\phi_{\tau}^1\frac{1-\lambda}{1-\alpha}\bar{S}_{\tau}+X_{\tau}^{\max,\bar{S}} \Big|B_+\Big]-\mathbb{E}^{\mathbb{Q}}[X_{\tau}^{\max,\bar{S}}|B_+]\\
&= \mathbb{E}^{\mathbb{Q}}\Big[\phi_{\tau}^0+\phi_{\tau}^1\frac{1-\lambda}{1-\alpha}\bar{S}_{\tau} \Big|B_+\Big]\leq \mathbb{E}^{\mathbb{Q}}\Big[\phi_{\tau}^0+\phi_{\tau}^{1}\frac{1-\lambda}{1-\alpha}S_{\tau}\Big|B_+\Big]\\
&< \mathbb{E}^{\mathbb{Q}}[-\hat{X}_{\tau}^{\max,\tilde{S}}|B_+]\leq \mathbb{E}^{\mathbb{Q}}[-\hat{X}_{T}^{\max,\tilde{S}}|B_+].\nonumber
\end{align*}
As either $\mathbb{P}(B_+)>0$ or $\mathbb{P}(B_{-})>0$, we arrive at a contradiction to $V(\phi^0,\phi^1)_T\geq -\hat{X}_T^{\max,\tilde{S}}$ $\mathbb{P}$-a.s., and our conclusion holds.
\end{proof}

\subsection{Proof of Theorem $\ref{mainthm}$}

The following proposition plays a central role to build a bipolar result required in the proof of Theorem $\ref{mainthm}$.

\begin{proposition}\label{dualprop}
Let Assumption $\ref{Assum}$, $\ref{assE}$ and $\ref{assNo}$ hold. The families $(\mathcal{C}(x,q))_{(x,q)\in\mathcal{K}}$ and $(\mathcal{D}(y,r))_{(y,r)\in\mathcal{L}}$ defined in $(\ref{primsp})$ and $(\ref{dualsp})$ have the following properties:
\begin{itemize}
\item[(i)] For any $(x,q)\in\mathcal{K}$, the set $\mathcal{C}(x,q)$ contains a strictly positive constant. A nonnegative function $g$ belongs to $\mathcal{C}(x,q)$ if and only if
\begin{equation}\label{bipo1}
\mathbb{E}[gh]\leq xy+q\cdot r,\ \ \text{for all}\ (y,r)\in\mathcal{L}\ \text{and}\ h\in\mathcal{D}(y,r).
\end{equation}
\item[(ii)] For any $(y,r)\in\mathcal{L}$, the set $\mathcal{D}(y,r)$ contains a strictly positive random variable. A nonnegative function $h$ belongs to $\mathcal{D}(y,r)$ if and only if
\begin{equation}\label{bipo2}
\mathbb{E}[gh]\leq xy+q\cdot r,\ \ \text{for all}\ (x,q)\in\mathcal{K}\ \text{and}\ g\in\mathcal{C}(x,q).
\end{equation}
\end{itemize}
\end{proposition}

The proof of Proposition $\ref{dualprop}$ is based on a sequel of auxiliary lemmas. Lemma $\ref{ineqC}$ together with Lemma $\ref{charac}$ below provide us a super-hedging result. In particular, the characterization of the set $\mathcal{C}(x,q)$ below gives one side of our super-hedging theorem.
\begin{lemma}\label{ineqC}
If $(x,q)\in\mathcal{K}$, for any $g\in\mathcal{C}(x,q)$, we have
\begin{equation}\label{gineq}
\mathbb{E}^\mathbb{Q}[g]\leq x+\mathbb{E}^{\mathbb{Q}}[q\cdot\mathcal{E}_T],\ \ \forall \mathbb{Q}\in\mathcal{M}.
\end{equation}
\end{lemma}
\begin{proof}
For any $g\in\mathcal{C}(x,q)$, there exits a $V\in\mathcal{H}(x,q)$, and $g\leq V_T+q\cdot\mathcal{E}_T$. It is hence enough to verify $\mathbb{E}^{\mathbb{Q}}[V_T+q\cdot\mathcal{E}_T]\leq x+\mathbb{E}^{\mathbb{Q}}[q\cdot\mathcal{E}_T]$, $\forall \mathbb{Q}\in\mathcal{M}$, which is equivalent to show that
\begin{equation}\label{ineq1}
\mathbb{E}^{\mathbb{Q}}[V_T+q\cdot\mathcal{E}_T]\leq x+\mathbb{E}^{\mathbb{Q}}[q\cdot\mathcal{E}_T],\ \ \forall \mathbb{Q}\in\mathcal{M}(\tilde{S}),\ \ \forall \tilde{S}\in\mathcal{S}.
\end{equation}
By the definition of acceptable portfolios, there exists a constant $a>0$ such that for each fixed $\tilde{S}\in\mathcal{S}$, there exists a maximal element $X^{\max,\tilde{S}}\in\mathcal{X}(\tilde{S}, a)$ with $V_{\tau}+X_{\tau}^{\max,\tilde{S}}\geq 0$ for all $[0,T]$-valued stopping times. Therefore we can
rewrite
\begin{equation}\label{rew}
\mathbb{E}^{\mathbb{Q}}[V_T+q\cdot\mathcal{E}_T]=\mathbb{E}^{\mathbb{Q}}[V_T+X_T^{\max,\tilde{S}}]-\mathbb{E}^{\mathbb{Q}}[X_T^{\max,\tilde{S}}]+\mathbb{E}^{\mathbb{Q}}[q\cdot\mathcal{E}_T].
\end{equation}
Define the set
\begin{equation}
\mathcal{M}'(\tilde{S})\triangleq \{\mathbb{Q}\in\mathcal{M}(\tilde{S}):\  X^{\max,\tilde{S}}\ \text{is a UI martingale under $\mathbb{Q}$}\},\nonumber
\end{equation}
Theorem $5.2$ of \cite{DS97} asserts that $\mathcal{M}'(\tilde{S})$ is not empty and dense in $\mathcal{M}(\tilde{S})$ with respect to the norm topology of $\mathbb{L}^1(\Omega, \mathcal{F}, \mathbb{P})$. We shall first verify that the inequality $(\ref{ineq1})$ holds for all $\mathbb{Q}\in\mathcal{M}'(\tilde{S})$. As $X_T^{\max,\tilde{S}}$ is a UI martingale under $\mathbb{Q}\in\mathcal{M}'(\tilde{S})$, it is sufficient to verify that
\begin{equation}\label{sinq}
\mathbb{E}^{\mathbb{Q}}[V_T+X_T^{\max,\tilde{S}}]\leq x+a,\ \ \forall \mathbb{Q}\in\mathcal{M}'(\tilde{S}).
\end{equation}
As $\tilde{S}\in[(1-\lambda)S, S]$, it is easy to see that $V_t\leq \tilde{V}_t$ where $\tilde{V}_t\triangleq \phi_t^0+\phi_t^1\tilde{S}_t$ for $t\in[0,T]$. Follow the proof of Proposition $1.6$ of \cite{Sch2}, we get that $\tilde{V}_t$ is a local optional strong supermartingale under each $\mathbb{Q}\in\mathcal{M}'(\tilde{S})$, therefore $\tilde{V}_t+X_t^{\max,\tilde{S}}$ is also a local optional strong supermartingale under $\mathbb{Q}$. Since $\tilde{V}_{t}+X_t^{\max,\tilde{S}}\geq V_t+X_t^{\max,\tilde{S}}\geq 0$, we can deduce that $\tilde{V}_t+X_t^{\max,\tilde{S}}$ is an optional strong supermartingale under $\mathbb{Q}$ by Fatou's Lemma. We obtain that
\begin{equation}
\mathbb{E}^{\mathbb{Q}}[\tilde{V}_T+X_T^{\max,\tilde{S}}]\leq x+a,\ \ \forall \mathbb{Q}\in\mathcal{M}'(\tilde{S}),\nonumber
\end{equation}
which implies that $(\ref{sinq})$ holds. Hence, it follows that for each $\tilde{S}\in\mathcal{S}$,
\begin{equation}\label{ineq2}
\mathbb{E}^{\mathbb{Q}}[V_T+q\cdot\mathcal{E}_T]\leq x+\mathbb{E}^{\mathbb{Q}}[q\cdot\mathcal{E}_T],\ \ \forall \mathbb{Q}\in\mathcal{M}'(\tilde{S}).
\end{equation}
Denote $\gamma_T\triangleq V_T+q\cdot\mathcal{E}_T$. The density property of $\mathcal{M}'(\tilde{S})$ in $\mathcal{\tilde{S}}$ in the norm topology of $\mathbb{L}^1$ implies the existence of a sequence of $\mathbb{Q}^n\in\mathcal{M}'(\tilde{S})$ and by $(\ref{ineq2})$, we have
\begin{align*}
\mathbb{E}^{\mathbb{Q}}[\gamma_T]&=\lim_{m\rightarrow\infty}\mathbb{E}^{\mathbb{Q}}[\gamma_T\mathbf{1}_{\{\gamma_T\leq m\}}]=\lim_{m\rightarrow\infty}\lim_{n\rightarrow\infty}\mathbb{E}^{\mathbb{Q}^n}[\gamma_T\mathbf{1}_{\{\gamma_T\leq m\}}]\\
&\leq \lim_{n\rightarrow\infty}\mathbb{E}^{\mathbb{Q}^n}[\gamma_T]\leq x+\lim_{n\rightarrow\infty}\mathbb{E}^{\mathbb{Q}^n}[q\cdot\mathcal{E}_T].
\end{align*}

Clearly for $m>0$ and each $1\leq i\leq N$, we have
\begin{equation}
\mathcal{E}^i_T\mathbf{1}_{\{\mathcal{E}^i_T>m\}}\leq \sum_{i=1}^{N}\mathcal{E}_T^i\mathbf{1}_{\{\sum_{i=1}^{N}\mathcal{E}_T^i >m\}},\ \ \mathbb{P}-\text{a.s.}
\end{equation}
The assumption that $\mathcal{E}^i_T\geq 0$ a.s. under $\mathbb{P}$ implies $\mathcal{E}^i_T\geq 0$ a.s. under $\mathbb{Q}\in\mathcal{M}$, it follows that
\begin{equation}\label{woo}
\lim_{m\rightarrow\infty}\sup_{\mathbb{Q}\in\mathcal{M}}\mathbb{E}^{\mathbb{Q}}[\mathcal{E}^i_T\mathbf{1}_{\{\mathcal{E}^i_T>m\}}]=0,\ \ 1\leq i\leq N.
\end{equation}

Given Assumption $\ref{assE}$, Moore-Osgood Theorem (see Theorem $5$, p.$102$ of \cite{13}) and Monotone Convergence Theorem give us that
\begin{align*}
\lim_{n\rightarrow\infty}\mathbb{E}^{\mathbb{Q}^n}[\mathcal{E}^i_T]=&\lim_{n\rightarrow\infty}\lim_{m\rightarrow\infty}\mathbb{E}^{\mathbb{Q}^n}[\mathcal{E}^i_T\mathbf{1}_{\{\mathcal{E}^i_T\leq m\}}]=\lim_{m\rightarrow\infty}\lim_{n\rightarrow\infty}\mathbb{E}^{\mathbb{Q}^n}[\mathcal{E}^i_T\mathbf{1}_{\{\mathcal{E}^i_T\leq m\}}]\\
=&\lim_{m\rightarrow\infty}\mathbb{E}^{\mathbb{Q}}[\mathcal{E}^i_T\mathbf{1}_{\{\mathcal{E}^i_T\leq m\}}]=\mathbb{E}^{\mathbb{Q}}[\mathcal{E}^i_T],\ \ 1\leq i\leq N. \nonumber
\end{align*}
We thereby obtain that $\lim_{n\rightarrow\infty}\mathbb{E}^{\mathbb{Q}^n}[q\cdot\mathcal{E}_T]=\mathbb{E}^{\mathbb{Q}}[q\cdot\mathcal{E}_T]$. It follows that $(\ref{ineq1})$ holds for any $\mathbb{Q}\in\mathcal{M}(\tilde{S})$ and any $\tilde{S}\in\mathcal{S}$, which completes the proof.
\end{proof}

For the other side of the super-hedging result, we need more delicate work. Fix a constant $\hat{a}>0$ and define $\mathcal{A}_{0,\hat{a}}$ as  the set of all pairs $\phi=(\phi^0, \phi^1)\in\mathcal{A}_0$ and for each $\tilde{S}\in\mathcal{S}$, there exits a $\hat{X}^{\max,\tilde{S}}\in\mathcal{X}(\tilde{S}, \hat{a})$ such that $V(\phi)_T+\hat{X}_T^{\max,\tilde{S}}\geq 0$. We intend to show that elements in the set  $\mathcal{A}_{0,\hat{a}}$ are bounded in probability. In fact, any convex combinations of the elements in $\mathcal{A}_{0,\hat{a}}$ are also bounded in probability. This is the first step to obtain the almost surely convergence result for any sequence in $\mathcal{A}_{0,\hat{a}}$ by passing to convex combinations.

\begin{lemma}\label{lemm}
Let $S$ and $0<\lambda<1$ satisfy the previous assumptions and suppose that $(CPS^{\lambda'})$ is satisfied in the local sense for some $0<\lambda'<\lambda$. For $\hat{a}>0$, we can find one probability measure $\mathbb{Q}\sim\mathbb{P}$ and there exist constants $C_0>0$ and $C_1>0$ such that for all $(\phi^0,\phi^1)\in \mathcal{A}_{0,\hat{a}}$, we have
\begin{equation}\label{pro1}
\mathbb{E}^{\mathbb{Q}}\Big[ \|\phi^0\|_T\Big]\leq C_0\hat{a},
\end{equation}
and
\begin{equation}\label{pro2}
\mathbb{E}^{\mathbb{Q}}\Big[ \|\phi^1\|_T\Big]\leq C_1\hat{a},
\end{equation}
where $\|\phi\|$ denotes the total variation of $\phi$.
\end{lemma}
\begin{proof}
Fix $0<\lambda'<\lambda$ as above.  Consider $\tilde{S}\in\mathcal{S}(\lambda',S)$ such that $\tilde{S}_t\in[(1-\lambda')S_t,S_t]$ and $(\tilde{S}_t)_{0\leq t\leq T}$ is a local $\mathbb{Q}$-martingale for all $\mathbb{Q}\in\mathcal{M}(\tilde{S};\lambda')$. Because the assertion of the lemma is of local type, we can assume by choosing stopping, that $\tilde{S}$ is a true martingale. We may also assume that $\phi_T^1=0$ so that the position in stock is liquidated at time $T$.

Assume $(\phi^0,\phi^1)$ is an acceptable portfolio under transaction costs $\lambda$ and $(\phi_0^0, \phi_0^1)=(0,0)$. Define the new process $\phi'=((\phi^0)',(\phi^1)')$ by
\begin{equation}
\phi'_t=((\phi^0)_t',(\phi^1)_t')=\Big(\phi_t^0+\frac{\lambda-\lambda'}{1-\lambda}\phi_t^{0,\uparrow}, \phi_t^1 \Big),\ \ 0\leq t\leq T.\nonumber
\end{equation}

Thanks to the proof of Lemma $3.1$ of \cite{Sch33}, $((\phi^0)',(\phi^1)')$ is a self-financing process under the transaction costs $\lambda'$. As $(\phi^0,\phi^1)$ is acceptable under transaction costs $\lambda$, and for any $\lambda'<\lambda$, it is clear that $\mathcal{S}(\lambda', S)\subset\mathcal{S}$. It follows that there exists a constant $a>0$ and for each $\tilde{S}\in\mathcal{S}(\lambda', S)$, there exists a $X^{\max,\tilde{S}}\in\mathcal{X}(\tilde{S}, a)$ such that $V(\phi)_{\tau}\geq -X_{\tau}^{\max,\tilde{S}}$ a.s. Moreover, it is easy to see that $V_{\tau}(\phi')\geq V_{\tau}(\phi)$ by the definition of $\phi'$. Therefore, we obtain that $\phi'=((\phi^0)',(\phi^1)')$ is an acceptable portfolio under the smaller transaction costs $\lambda'$.

Following the proof of Proposition $1.6$ of \cite{Sch33}, we see that $V(\phi')_t\leq \tilde{V}(\phi')_t$ where $\tilde{V}(\phi')_t\triangleq (\phi^0)'_t+(\phi^1)'_t\tilde{S}$, $\tilde{S}\in\mathcal{S}(\lambda', S)$ and $\tilde{V}(\phi')$ is a local optional strong super-martingale under all $\mathbb{Q}\in\mathcal{M}(\tilde{S};\lambda')$. For each fixed $\tilde{S}\in\mathcal{S}(\lambda',S)$, by the definition of acceptable portfolio, there exists a constant $a$ and a maximal element $X^{\max,\tilde{S}}\in\mathcal{X}(\tilde{S}, a)$ such that $V(\phi')_t\geq -X^{\max,\tilde{S}}_t$. Hence, we get $\tilde{V}(\phi')_t+X_t^{\max,\tilde{S}}\geq 0$ is an optional strong supermartingale. For this fixed $X^{\max,\tilde{S}}$, consider the set
\begin{equation}
\mathcal{M}'(\tilde{S}; \lambda')\triangleq \{\mathbb{Q}\in\mathcal{M}(\tilde{S}; \lambda'):\  X^{\max,\tilde{S}}\ \text{is a UI martingale under $\mathbb{Q}$}\}.\nonumber
\end{equation}
For each $\mathbb{Q}\in\mathcal{M}'(\tilde{S} ; \lambda')$, we obtain that
\begin{align*}
\mathbb{E}^{\mathbb{Q}}[(\phi^0)'_T+(\phi^1)'_T\tilde{S}_T]=&\mathbb{E}^{\mathbb{Q}}[(\phi^0)'_T+(\phi^1)'_T\tilde{S}_T+X_T^{\max,\tilde{S}}]-\mathbb{E}^{\mathbb{Q}}[X_T^{\max,\tilde{S}}]\\
\leq& 0+a-a=0.\nonumber
\end{align*}
By the definition of $(\phi^0)'$ and $(\phi^1)'$, we deduce that
\begin{equation}
\mathbb{E}^{\mathbb{Q}}[\phi_T^0+\phi_T^1\tilde{S}_T]+\frac{\lambda-\lambda'}{1-\lambda}\mathbb{E}^{\mathbb{Q}}[\phi_T^{0,\uparrow}]\leq 0.\nonumber
\end{equation}
Because $\phi^0_T+\phi_T^1\tilde{S}_T\geq V(\phi)_T$ and $\tilde{S}\in\mathcal{S}$, by definition, there exists a constant $\hat{a}>0$ and $\hat{X}^{\max,\tilde{S}}$ such that $V(\phi)_T\geq -\hat{X}_T^{\max,\tilde{S}}$. We obtain that
\begin{equation}
\mathbb{E}^{\mathbb{Q}}[\phi_T^{0,\uparrow}]\leq \frac{1-\lambda}{\lambda-\lambda'}\mathbb{E}^{\mathbb{Q}}[\hat{X}_T^{\max,\tilde{S}}]\leq \frac{(1-\lambda)\hat{a}}{\lambda-\lambda'},\ \ \forall \mathbb{Q}\in\mathcal{M}'(\tilde{S};\lambda')\nonumber
\end{equation}
as $\hat{X}^{\max,\tilde{S}}$ is a supermartingale under $\mathbb{Q}\in\mathcal{M}'(\tilde{S}; \lambda')$. As the set $\mathcal{M}'(\tilde{S}; \lambda')$ is dense in $\mathcal{M}(\tilde{S}; \lambda')$ with respect to the norm topology of $\mathbb{L}^1$, for any $\mathbb{Q}\in\mathcal{M}(\tilde{S}; \lambda')$, Fatou's lemma leads to
\begin{equation}
\mathbb{E}^{\mathbb{Q}}[\phi_T^{0,\uparrow}]\leq \frac{1-\lambda}{\lambda-\lambda'}\hat{a}.\nonumber
\end{equation}
For each $\tilde{S}\in\mathcal{S}(\lambda')$, $\phi_T^0=\phi_T^0+\phi_T^1\tilde{S}\geq -\hat{X}^{\max,\tilde{S}}$ by the previous argument and $\phi_T^1=0$. Therefore, it follows that
$\phi_T^{0,\downarrow}\leq \phi_T^{0,\uparrow}+\hat{X}_T^{\max,\tilde{S}}$. For each $\mathbb{Q}\in\mathcal{M}(\tilde{S}; \lambda')$, as $X^{\max,\tilde{S}}$ is a supermartingale under $\mathbb{Q}$, we can derive that
\begin{equation}
\mathbb{E}^{\mathbb{Q}}[\phi_T^{0,\uparrow}+\phi_T^{0,\downarrow}]\leq 2\frac{1-\lambda}{\lambda-\lambda'}\hat{a}+\hat{a},\nonumber
\end{equation}
which completes the proof of $(\ref{pro1})$.

As regards $(\ref{pro2})$, we can follow the proof of Lemma $3.1$ of \cite{Sch33}. First, we have that
\begin{equation}\label{phi1}
d\phi_t^{1,\uparrow}\leq \frac{d\phi_t^{0,\downarrow}}{S_t}.
\end{equation}
As $\tilde{S}$ is a $\mathbb{Q}$-local supermartingale and it follows that it is a $\mathbb{Q}$-supermartingale for $\mathbb{Q}\in\mathcal{M}(\tilde{S}; \lambda')$. It is easy to see that $\inf_{0\leq t\leq T}\tilde{S}_t(\omega)$ is $\mathbb{Q}$-a.s. strictly positive as $\tilde{S}_T>0$, $\mathbb{Q}$-a.s.. Therefore, $\inf_{0\leq t\leq T}\tilde{S}_t(\omega)$ is $\mathbb{P}$-a.s. as well. We can obtain that for any $\epsilon>0$, there exists $\delta>0$ such that
\begin{equation}\label{phi12}
\mathbb{P}\Big[ \inf_{0\leq t\leq T}S_t<\delta\Big]<\frac{\epsilon}{2}.
\end{equation}
Combining $(\ref{pro1})$, $(\ref{phi1})$ and $(\ref{phi12})$, it is easy to derive a control such that $\mathbb{E}^{\mathbb{Q}}[\phi_T^{1,\uparrow}]\leq k\hat{a}$ for some $k>0$. Finally, we recall that $\phi_T^1=0$ which implies that $\phi_T^{1,\uparrow}=\phi_T^{1,\downarrow}$. It follows that $(\ref{pro2})$ holds.
\end{proof}

It is now important for us to verify the closedness property of the set $\mathcal{U}_x$ (resp. $\mathcal{V}_x$) for the purpose of the super-hedging result. In particular, it is enough to consider the case $x=0$. For the admissible portfolio processes in Definition $\ref{admissiblephi}$, the Fatou-closedness is an appropriate concept, see Appendix $5.5$ of \cite{KS09}. As in \cite{Sch33}, a sequence $(\phi_T^n)_{n=1}^{\infty}$ in $\mathbb{L}^0(\mathbb{R}^2)$ \textit{Fatou-converges} to $\phi_T\in\mathbb{L}^0(\mathbb{R}^2)$ if there is $M>0$ such that $V(\phi^{0,n},\phi^{1,n})_T\geq -M$ and $\phi_T^n$ converges a.s. to $\phi_T$. A set is Fatou closed if it is closed under the Fatou convergence. Due to the stochastic lower-bounds for our acceptable portfolios, the previous Fatou-closedness has to be modified using the following alternative definition.

\begin{definition}
Fix some $\hat{a}>0$, for each $\tilde{S}\in\mathcal{S}$, we pick and fix one maximal element $\hat{X}^{\max,\tilde{S}}\in\mathcal{X}(\tilde{S},\hat{a})$. The set $\mathcal{U}_0$ (\text{resp.} $\mathcal{V}_0$) is said to be \textbf{relatively Fatou closed} if for any sequence $(\phi^{0,n}_T,\phi^{1,n}_T)\in\mathcal{U}_0$ which satisfies $V(\phi^n)_T\geq -\hat{X}_T^{\max,\tilde{S}}$, $\tilde{S}\in\mathcal{S}$ (\text{resp.} $\phi_T^{0,n}\geq -\hat{X}_T^{\max,\tilde{S}}$, $\tilde{S}\in\mathcal{S}$) and converges to $(\phi^0_T,\phi^1_T)\in\mathbb{L}^0(\mathbb{R}^2)$ (\text{resp.} $\phi_T^0\in\mathbb{L}^0(\mathbb{R})$) almost surely, we have that $(\phi_T^0,\phi_T^1)\in\mathcal{U}_0$ (\text{resp.} $\phi_T^0\in\mathcal{V}_0$).
\end{definition}

\begin{remark}\label{order}
In the two dimensional setting, let us introduce the partial order on $\mathbb{L}^0(\mathbb{R}^2)$ by $(\phi^0,\phi^1)\succeq (\psi^0, \psi^1)$ if $V(\phi^0-\psi^0, \phi^1-\psi^1)_T\geq 0$, a.s.. Therefore, in the above definition, we can also say $(\phi^{0,n}, \phi^{1,n})$ relatively Fatou converges to $(\phi^0,\phi^1)$, if there exists a constant $\hat{a}>0$ such that for each $\tilde{S}\in\mathcal{S}$, we can find one $\hat{X}^{\max,\tilde{S}}\in\mathcal{X}(\tilde{S}, \hat{a})$ such that $(\phi_T^{0,n}, \phi_T^{1,n})\succeq (-\hat{X}_T^{\max,\tilde{S}}, 0)$ and $(\phi_T^{0,n},\phi_T^{1,n})$ converges to $(\phi^0_T, \phi^1_T)$ almost surely.
\end{remark}

\begin{lemma}\label{fatoulem}
Fix $S=(S_t)_{0\leq t\leq T}$ and $0<\lambda<1$ as above and let Assumption $\ref{Assum}$ hold. The sets $\mathcal{U}_0$ and $\mathcal{V}_0$ are both relatively Fatou closed.
\end{lemma}
\begin{proof}
Fix $\hat{a}>0$ and for each $\tilde{S}\in\mathcal{S}$, choose and fix $\hat{X}^{\max,\tilde{S}}\in\mathcal{X}(\tilde{S}, \hat{a})$. Consider a sequence $(\phi_T^{0,n},\phi_T^{1,n})\in\mathcal{U}_0$ such that $V(\phi^n)_T\geq -\hat{X}_T^{\max,\tilde{S}}$ and $(\phi_T^{0,n},\phi_T^{1,n})$ converges a.s. to some $(\phi_T^0,\phi_T^1)\in\mathbb{L}^0(\mathbb{R}^2)$. Thanks to Proposition $\ref{auxlem}$, we can deduce that for any $[0,T]$-valued stopping time $\tau$, $V(\phi^n)_{\tau}\geq -\hat{X}_{\tau}^{\max,\tilde{S}}$. Decompose canonically these processes $\phi_T^{0,n}=\phi_T^{0,n,\uparrow}-\phi_T^{0,n,\downarrow}$ and $\phi_T^{1,n}=\phi_T^{1,n,\uparrow}-\phi_T^{1,n,\downarrow}$. Thanks to Lemma $\ref{lemm}$, the proof of Theorem $3.4$ of \cite{Sch33} can be carried over verbatim in our setting, and we can find a predictable increasing process $\phi^{0,\uparrow}=(\phi^{0,\uparrow}_t)_{0\leq t\leq T}$ such that the sequence $\phi_t^{0,n,\uparrow}$ converges almost surely to $\phi_t^{0,\uparrow}$ for all $0\leq t\leq T$. Similar results hold for $\phi^{0,\downarrow}$, $\phi^{1,\uparrow}$ and $\phi^{1,\downarrow}$. These processes are all predictable, increasing and satisfy condition $(\ref{self})$.

Define the process $(\phi^0_t,\phi^1_t)_{0\leq t\leq T}$ by $\phi_t^0=\phi_t^{0,\uparrow}-\phi_t^{0,\downarrow}$ and $\phi_t^1=\phi_t^{1,\uparrow}-\phi_t^{1,\downarrow}$. The process $\phi$ is a predictable and self-financing portfolio process. Moreover, for each $\tilde{S}\in\mathcal{S}$, as $V(\phi^n)_{\tau}\geq -\hat{X}_{\tau}^{\max,\tilde{S}}$ for all $[0,T]$-valued stopping time $\tau$, we obtain that $V(\phi)_{\tau}\geq -\hat{X}_{\tau}^{\max,\tilde{S}}$ as the convergence of $(\phi^n)_{n=1}^{\infty}$ takes place for all $t\in[0,T]$. We can conclude that $(\phi^0,\phi^1)$ is an acceptable portfolio, i.e., $(\phi^0,\phi^1)\in\mathcal{U}_0$, and therefore $\mathcal{U}_0$ is relatively Fatou closed. The proof for $\mathcal{V}_0$ follows the same arguments.
\end{proof}

After the closedness property, we need to proceed to characterize the auxiliary set $\mathcal{W}(x; \lambda, S)$ (short as $\mathcal{W}(x)$) of two dimensional random variables for the purpose of the super-hedging result, where we define
\begin{align*}
\mathcal{W}(x)=\{(W^0, W^1): &W^0=\phi_T^0+\essinf_{\tilde{S}\in\mathcal{S}}X_T^{\text{max},\tilde{S}}, W^1=\phi_T^1,\ \text{for}\ (\phi^0_T,\phi^1_T)\in \mathcal{U}_x\\
&\text{with its corresponding thresholds}\ X^{\text{max},\tilde{S}}\\
&\text{in the definition of acceptable portfolios}\}.
\end{align*}

Furthermore, let us consider the auxiliary set $\mathcal{W}^{\infty}(x)$ of bounded random variables as elements in the set $\mathcal{W}(x)$ in the sense that $\mathcal{W}^{\infty}(x)=\mathcal{W}(x)\cap\mathbb{L}^{\infty}$.

\begin{definition}\label{nameZ}
Denote $\bar{\mathcal{Z}}(\lambda, S)$ (short as $\bar{\mathcal{Z}}$) as the set of all pairs $Z_T=(Z^0_T, Z^1_T)\in\mathbb{L}^1_+(\mathbb{R}^2;\mathcal{F}_T)$ such that $\mathbb{E}[Z_T^0]=1$ and
\begin{equation}\label{lateineq}
\mathbb{E}[W^0Z_T^0+W^1Z_T^1]\leq x+\mathbb{E}\left[\essinf_{\tilde{S}\in\mathcal{S}}X_T^{\text{max},\tilde{S}}Z_T^0\right],
\end{equation}
for all $(W^0, W^1)\in\mathcal{W}^{\infty}(x)$.
\end{definition}

Each $(Z_T^0, Z_T^1)\in\bar{\mathcal{Z}}$ can be identified with a pair $(\mathbb{Q}, \tilde{S})$ by setting
\begin{equation}
Z^i_t=\mathbb{E}[Z^i_T|\mathcal{F}_t],\ \ i=0,1,\ \ \tilde{S}_t=\frac{Z_t^1}{Z_t^0},\ \ \text{and}\ \ \ \frac{d\mathbb{Q}}{d\mathbb{P}}=Z_T^0.\nonumber
\end{equation}
However, here the measure $\mathbb{Q}$ is only absolutely continuous with respect to $\mathbb{P}$.

The following lemma builds the relationship between the definition of $\bar{\mathcal{Z}}$ and $\lambda$-CPS and shows that the set $\bar{\mathcal{Z}}$ is actually independent of the choice of the random endowments $q\cdot\mathcal{E}_T$.

\begin{lemma}\label{lemZ}
Assume that $S_t\leq K$ for some constant $K$ for all $0\leq t\leq T$. For each $Z\in\bar{\mathcal{Z}}$, define the martingale $Z=(Z_t^0, Z_t^1)_{0\leq t\leq T}$ by
\begin{equation}
Z^i_t\triangleq \mathbb{E}[Z^i_T|\mathcal{F}_t],\ \ i=0,1,\ \ 0\leq t\leq T.\nonumber
\end{equation}
We will have that
\begin{equation}
\tilde{S}_t\triangleq \frac{Z_t^1}{Z_t^0}\in [(1-\lambda)S_t,S_t],\ \ 0\leq t\leq T,\ \text{a.s.}\nonumber
\end{equation}
Conversely, suppose that $Z=(Z_t^0,Z_t^1)_{0\leq t\leq T}$ is an $\mathbb{R}_+^2$-valued $\mathbb{P}$-martingale such that $Z_0^0=1$ and $\tilde{S}_t=\frac{Z_t^1}{Z_t^0}$ takes values in $[(1-\lambda)S_t, S_t]$ a.s. on $\{Z_t^0>0\}$. Then we have $Z_T=(Z_T^0,Z_T^1)\in\bar{\mathcal{Z}}$.
\end{lemma}
\begin{proof}
Choose any $Z_T\in\bar{\mathcal{Z}}$ and suppose that there exits a $[0,T)$-valued stopping time $\tau$ such that $\mathbb{Q}(\tilde{S}_{\tau} >S_{\tau})>0$. Let us consider the strategy
\begin{equation}
a_t=\Big(-1, \frac{1}{S_{\tau}}\Big)\mathbf{1}_{\{\tilde{S}_{\tau}>S_{\tau}\}}\mathbf{1}_{\rrbracket \tau, T\rrbracket}(t),\ \ 0\leq t\leq T,
\end{equation}
It is clear that $a_t=(\phi^0_t,\phi^1_t)\in\mathcal{U}_0$ is a self-financing strategy for $t\in[0,T]$. Moreover, it is also clear that for each $\tilde{S}\in\mathcal{S}$, we can choose $X^{\max,\tilde{S}}\equiv 1$ such that $W^0=\phi^0_T+X^{\max,\tilde{S}}=\phi_T^0+1\in\mathbb{L}^{\infty}$ and $W^1=\phi_T^1\in\mathbb{L}^{\infty}$ such and $(W^0, W^1)\in\mathcal{W}^{\infty}(0)$. Using the fact that $\tilde{S}$ is a martingale, we can follow the proof of Proposition $4.2$ of \cite{Sch33} to deduce a contradiction. To wit, we calculate that
\begin{align*}
&\mathbb{E}^{\mathbb{P}}[W^0Z_T^0+W^1Z_T^1]\\
=&\mathbb{E}^{\mathbb{P}}\left[ \left( -Z_T^0+\frac{Z_T^1}{S_{\tau}}\right)\mathbf{1}_{\{\tilde{S}_{\tau}>S_{\tau}\}}\right]+1=\mathbb{E}^{\mathbb{P}}\left[ \mathbb{E}^{\mathbb{P}}\left[ \left(-Z_T^0+\frac{Z_T^1}{S_{\tau}}\right)\mathbf{1}_{\{\tilde{S}_{\tau}>S_{\tau}\}} \Big|\mathcal{F}_{\tau}\right] \right]+1\\
=&\mathbb{E}^{\mathbb{P}}\left[ Z_{\tau}^0\left( -1+\frac{\tilde{S}_{\tau}}{S_{\tau}}\right)\mathbf{1}_{\{\tilde{S}_{\tau}>S_{\tau}\}}\right]+1=\mathbb{E}^{\mathbb{Q}}\left[\left( -1+\frac{\tilde{S}_{\tau}}{S_{\tau}}\right)\mathbf{1}_{\{\tilde{S}_{\tau}>S_{\tau}\}}\right]+1>1,
\end{align*}
which is a contradiction to $(\ref{lateineq})$. If $\mathbb{Q}(\tilde{S}_T>S_T)>0$, we can instead consider the portfolio process $a'_t=\Big(-1, \frac{1}{S_T}\Big)\mathbf{1}_{\{\tilde{S}_T>S_T\}}\mathbf{1}_{\llbracket T\rrbracket}$ and deduce a similar contradiction and therefore $\tilde{S}_t\leq S_t$ for all $0\leq t\leq T$.

As given in the proof of Proposition $4.2$ of \cite{Sch33}, the strategies
\begin{equation}
b_t=((1-\lambda)S_{\tau}, -1)\mathbf{1}_{\{\tilde{S}_{\tau}<(1-\lambda)S_{\tau}\}}\mathbf{1}_{\rrbracket \tau,T\rrbracket}(t),\ \ 0\leq t\leq T, \nonumber
\end{equation}
and
\begin{equation}
b'_t=((1-\lambda)S_{T}, -1)\mathbf{1}_{\{\tilde{S}_{T}<(1-\lambda)S_{T}\}}\mathbf{1}_{\llbracket T\rrbracket}(t),\ \ 0\leq t\leq T \nonumber
\end{equation}
satisfy $b_T\in\mathcal{U}_0$ (resp. $b'_T\in\mathcal{U}_0$). Notice that $V(b)_t\geq -K$ (resp. $V(b')_t\geq -K$) for $t\in[0,T]$, it is enough to choose that $X^{\max,\tilde{S}}=K$ for all $\tilde{S}\in\mathcal{S}$. By picking $W^0=\phi_T^0+K\in\mathbb{L}^{\infty}$ and $W^1=\phi_T^1\in\mathbb{L}^{\infty}$, we get that $(W^0, W^1)\in\mathcal{W}^{\infty}(0)$. Following the previous proof again, we can derive that $\tilde{S}_t\geq (1-\lambda)S_t$ for all $0\leq t\leq T$ using the above constructions of portfolios $(b_t)_{0\leq t\leq T}$ and $(b'_t)_{0\leq t\leq T}$.

For the other direction, for any $(W^0, W^1)\in\mathcal{W}^{\infty}(0)$, we have
\begin{align*}
W^0Z_T^0+W^1Z_T^1&=(\phi_T^0+\essinf_{\tilde{S}\in\mathcal{S}}X_T^{\text{max},\tilde{S}} )Z_T^0+\phi_T^1Z_T^1\\
&=(\phi_T^0+\phi_T^1\tilde{S}_T^{\ast}+ \essinf_{\tilde{S}\in\mathcal{S}}X_T^{\text{max},\tilde{S}})Z_T^0,\nonumber
\end{align*}
where $(\mathbb{Q}^{\ast}, \tilde{S}^{\ast})$ is the pair induced by $(Z^0_T, Z^1_T)$. Define the random variable $\mathcal{E}=\essinf_{\tilde{S}\in\mathcal{S}}X_T^{\text{max},\tilde{S}})Z_T^0$. It follows that
\begin{equation}
\mathbb{E}\left[W^0Z^0_T+W^1Z^1_T\right]\leq \mathbb{E}^{\mathbb{Q}^{\ast}}\left[ \phi_T^0+\phi_T^1\tilde{S}^{\ast}_T+X_T^{\text{max},\tilde{S}^{\ast}}\right]-\mathbb{E}^{\mathbb{Q}^{\ast}}\left[X_T^{\max, \tilde{S}^{\ast}}\right]+\mathbb{E}^{\mathbb{Q}^{\ast}}\left[\mathcal{E}\right].
\end{equation}
It is safe to split the above integral because $W^0\in\mathbb{L}^{\infty}$, $W^1\in\mathbb{L}^{\infty}$ and both $X_T^{\max, \tilde{S}^{\ast}}\geq 0$ and $\mathcal{E}\geq 0$ $\mathbb{P}$-a.s. as well as $\mathbb{Q}^{\ast}$-a.s.. Therefore, the expectation $\mathbb{E}^{\mathbb{Q}^{\ast}}\left[ \phi_T^0+\phi_T^1\tilde{S}^{\ast}_T+X_T^{\text{max},\tilde{S}^{\ast}}\right]$ is well defined. We can simply mimic the proof of Lemma $\ref{ineqC}$ and obtain that the validity of $(\ref{lateineq})$ which completes the proof that $(Z_T^0, Z_T^1)\in\bar{\mathcal{Z}}$.
\end{proof}

We now pass from the auxiliary set $\mathcal{W}^{\infty}(0)$ to the set $\mathcal{W}(x)$ and characterize the set $\mathcal{W}(x)$ still using the same dual set $\bar{\mathcal{Z}}$.

\begin{lemma}\label{lem13}
Assume that $S_t\leq K$ for some constant $K>0$, $0\leq t\leq T$. We can characterize the set $\mathcal{W}(x)$ using the set $\bar{\mathcal{Z}}$ by
\begin{align}\label{chaUacp}
\mathcal{W}(x)=\Big\{ (W^0, W^1)\in\mathbb{L}^0_C(\mathbb{R}^2): &\mathbb{E}[W^0Z_T^0+W^1Z_T^1]\leq x+\nonumber\\
&\mathbb{E}[\essinf_{\tilde{S}\in\mathcal{S}}X_T^{\max,\tilde{S}}Z_T^0],\ \ \forall Z_T\in\bar{\mathcal{Z}}\Big\},
\end{align}
where $(a,b)\in\mathbb{L}^0_C(\mathbb{R}^2)$ satisfies $(a,b)\succeq (0,0)$.
\end{lemma}
\begin{proof}
In the two dimensional setting with partial order defined in Remark $\ref{order}$, for any constant $\kappa>0$, it is easy to verify that the intersection of $\mathcal{W}^{\infty}(0)$ with the ball $\{\xi:\|\xi\|_{\infty}\leq \kappa\}$ is closed in probability. Proposition $5.5.1$ of \cite{KS09} gives that $\mathcal{U}^{\infty}(0)$ is $\text{weak}^{\ast}$ closed (i.e., closed in $\sigma(\mathbb{L}^{\infty}, \mathbb{L}^1)$). It is shown in Theorem $5.5.3$ of \cite{KS09} that we have the following characterization
\begin{align*}
\mathcal{W}^{\infty}(0)=\Big\{(W^0, W^1)\in\mathbb{L}^{\infty}_C(\mathbb{R}^2):
&\mathbb{E}[W^0Z_T^0+W^1Z_T^1]\leq \mathbb{E}[\essinf_{\tilde{S}\in\mathcal{S}}X_T^{\max,\tilde{S}}Z_T^0],\\
&\forall Z\in(\mathcal{U}_0^{\infty})^{\circ}\Big\}.\nonumber
\end{align*}
As $\mathcal{W}{\infty}(0)$ contains the negative orthant $-\mathbb{L}^{\infty}(\mathbb{R}^2)$, its polar $(\mathcal{W}^{\infty}(0))^{\circ}$ is therefore defined by
\begin{align*}
(\mathcal{W}^{\infty}(0))^{\circ}=\Big\{(Z_T^0,Z_T^1)\in\mathbb{L}^1(\mathbb{R}^2): &\mathbb{E}[W^0Z_T^0+W^1Z_T^1]\leq \mathbb{E}[\essinf_{\tilde{S}\in\mathcal{S}}X_T^{\max,\tilde{S}}Z_T^0],\\
&\ \forall (W^0,W^1)\in\mathcal{W}^{\infty}(0)\Big\}.
\end{align*}
By Definition $\ref{nameZ}$ and Lemma $\ref{lemZ}$, it is clear that $\bar{\mathcal{Z}}=(\mathcal{W}^{\infty}(0))^{\circ}$, and hence, we get that
\begin{align}\label{chaU}
\mathcal{W}^{\infty}(0)=\Big\{(W^0, W^1)\in\mathbb{L}^{\infty}_C(\mathbb{R}^2): &\mathbb{E}[W^0Z_T^0+W^1Z_T^1]\leq \mathbb{E}[\essinf_{\tilde{S}\in\mathcal{S}}X_T^{\max,\tilde{S}}Z_T^0],\nonumber\\
&\forall Z\in\bar{\mathcal{Z}}\Big\}.
\end{align}

We then claim that $\mathcal{W}^{\infty}(0)$ is relatively Fatou dense in $\mathcal{W}(0)$. To wit, let us consider any $(\phi_T^0,\phi_T^1)\in\mathcal{U}_0$ with the existence of $\hat{a}>0$, for each $\tilde{S}\in\mathcal{S}$, there exits a $\hat{X}^{\max,\tilde{S}}\in\mathcal{X}(\tilde{S}, \hat{a})$ such that $V(\phi^0,\phi^1)_T+\hat{X}_T^{\max,\tilde{S}}\geq 0$. We need to show that there exits a sequence $(W^{0,n}, W^{1,n})\in\mathcal{W}^{\infty}(0)$ such that $(W^{0,n}, W_T^{1,n})\rightarrow (W^0, W^1)$ a.s.

Define the set
\begin{equation}
E_n\triangleq \left\{|V_T(\phi^0,\phi^1)+\essinf_{\tilde{S}\in\mathcal{S}}\hat{X}_T^{\text{max},\tilde{S}}|\leq n, |\phi_T^1|\leq n\right\},
\end{equation}
and denote $E_n^c$ the complement of the set $E_n$.

Define the sequence $(\phi^{0,n}_T, \phi^{1,n}_T)$ by
\begin{equation}
\phi_T^{0,n}\triangleq \phi_T^0\mathbf{1}_{E_n}-\essinf_{\tilde{S}\in\mathcal{S}}\hat{X}_T^{\max,\tilde{S}}\mathbf{1}_{E_n^c},\ \ \ \phi_T^{1,n}\triangleq \phi_T^1\mathbf{1}_{E_n}.\nonumber
\end{equation}
For $0\leq t<T$, let us choose $\phi_t^{0,n}=\phi_t^0$ and $\phi_t^{1,n}=\phi_t^1$. It follows that $(\phi_t^{0,n}, \phi_t^{1,n})$ is a self-financing portfolio. Indeed, it is enough to check the terminal time $T$.
If $E_n^c$ happens, we close the position by liquidation.

We then define the sequence
\begin{equation}
W^{0,n}\triangleq \phi_T^{0,n}+\essinf_{\tilde{S}\in\mathcal{S}}\hat{X}_T^{\max,\tilde{S}},\ \ \ W^1\triangleq \phi_T^{1,n},\nonumber
\end{equation}
and
\begin{equation}
W^{0}\triangleq \phi_T^{0}+\essinf_{\tilde{S}\in\mathcal{S}}\hat{X}_T^{\max,\tilde{S}},\ \ \ W^1\triangleq \phi_T^{1}.\nonumber
\end{equation}

Clearly, $W^{i,n}\rightarrow W^{i}$ a.s. for $i=0,1$ as $(\phi_T^{0,n}+\essinf_{\tilde{S}\in\mathcal{S}}\hat{X}_T^{\max,\tilde{S}}, \phi_T^{1,n})\rightarrow (\phi^0_T+\essinf_{\tilde{S}\in\mathcal{S}}\hat{X}_T^{\max,\tilde{S}}, \phi_T^1)$ a.s.

Moreover, it follows by definition that $(\phi_T^{0,n}, \phi_T^{1,n})\succeq (-\hat{X}_T^{\max,\tilde{S}}, 0)$ as we have
\begin{align*}
V(\phi^{0,n}, \phi^{1,n})_T+\hat{X}_T^{\max,\tilde{S}}&\geq V(\phi^0,\phi^1)_T\mathbf{1}_{E_n}-\hat{X}_T^{\max,\tilde{S}}\mathbf{1}_{E_n^c}  +\hat{X}_T^{\max,\tilde{S}}\\
&=V(\phi^0,\phi^1)_T\mathbf{1}_{E_n}+\hat{X}_T^{\max,\tilde{S}}\mathbf{1}_{E_n}\\
&=\Big(V(\phi^0,\phi^1)_T + \hat{X}_T^{\max,\tilde{S}}\Big)\mathbf{1}_{E_n}\geq 0.
\end{align*}
Therefore, $(W^{0,n}, W^{1,n})\succeq (0,0)$ for each $n$. In addition, we also have
\begin{equation}
V(\phi^{0,n}, \phi^{1,n})_T+\essinf_{\tilde{S}\in\mathcal{S}}\hat{X}_T^{\max,\tilde{S}}=\Big(V(\phi^0,\phi^1)_T + \essinf_{\tilde{S}\in\mathcal{S}}\hat{X}_T^{\max,\tilde{S}}\Big)\mathbf{1}_{E_n}\in\mathbb{L}^{\infty}.\nonumber
\end{equation}
It yields that $(W^{0,n}, W^{1,n})\in\mathbb{L}^{\infty}_C(\mathbb{R}^2)$ which implies that the claim holds.

By using $(\ref{chaU})$ and the fact that $\mathcal{W}^{\infty}(x)$ is relatively Fatou dense in $\mathcal{W}(x)$, it is straightforward to verify the characterization in $(\ref{chaUacp})$.
\end{proof}

Based on all previous results from Lemma $\ref{lemm}$ to Lemma $\ref{lem13}$, we can finally build the other side of the super-hedging theorem for acceptable portfolios. The proof relies heavily on the characterization $(\ref{chaUacp})$ in Lemma $\ref{lem13}$. As it is assumed that the price process $S$ is uniformly bounded in Lemma $\ref{lem13}$, the trick of working with the localizing sequence becomes necessary here.

\begin{lemma}\label{charac}
Fix some $x\in\mathbb{R}$. Let $g$ be an $\mathbb{R}$-valued, $\mathcal{F}_T$-measurable random variable such that there exits a constant $a>0$ and for each $\tilde{S}\in\mathcal{S}(\lambda)$, there exists a $X^{\max,\tilde{S}}\in\mathcal{X}(\tilde{S}, a)$ with $g+X_T^{\max,\tilde{S}}\geq 0$. If for each $\lambda$-CPS $(\mathbb{Q},\tilde{S})$, i.e., for all $\mathbb{Q}\in\mathcal{M}$, we have
\begin{equation}\label{ineassum}
\mathbb{E}^{\mathbb{Q}}[g]\leq x,
\end{equation}
then there exits a pair $(\phi^0, \phi^1)\in\mathcal{A}_x$ such that $(\phi^0_0, \phi_0^1)=(x,0)$ and $(\phi_T^0, \phi_T^1)=(g,0)$.
\end{lemma}
\begin{proof}

As $S$ is locally bounded, let us consider the localizing sequence $(\tau_n)_{n\in\mathbb{N}}$ such that $S_{t\wedge \tau_n}\leq K(n)$. Define
\begin{equation}
g^n=\left\{
\begin{array}{rl}
g,&\ \ \ \ \text{on}\ \ \{\tau_n=T\},\\
\underset{\mathbb{Q}\in\mathcal{M}(\lambda, S)}{\text{essinf}}\mathbb{E}^{\mathbb{Q}}\bigg[-\underset{\tilde{S}\in\mathcal{S}(\lambda)}{\text{essinf}}\ \Big(X_T^{\max,\tilde{S}}\Big) \bigg|\mathcal{F}_{\tau_n}\bigg],&\ \ \ \ \text{on}\ \ \{\tau_n<T\}.\nonumber
\end{array}\right.
\end{equation}
It is clear that $(g^n)_{n=1}^{\infty}$ is $\mathcal{F}_{\tau_n}$-measurable and $g^n$ converges to $g$, $\mathbb{P}$-a.s..

Let $0<\lambda_n <\lambda$ be a sequence of real numbers increasing to $\lambda$. For each fixed $n\in\mathbb{N}$,  we consider the stopped process $S^{\tau_n}$ with the transaction costs $\lambda^n$. It is easy to check that for $0<\lambda'<1$, any stopped $\lambda'$-CPS $(\mathbb{Q}, \tilde{S}^{\tau_n})$ for $S$ is also a $\lambda'$-CPS for $S^{\tau_n}$. Moreover, by Proposition $6.1$ of \cite{Sch33}, for any stopped $\lambda'$-CPS $(\mathbb{Q},\tilde{S}^{\tau_n})$, we obtain that $\tilde{S}^{\tau_n}$ is a true $\mathbb{Q}$-martingale instead of a $\mathbb{Q}$-local martingale. Therefore, for any $0<\lambda'<1$, the stopped process $S^{\tau_n}$ admits a $\lambda'$-CPS $(\mathbb{Q}, \tilde{S})$ such that $\tilde{S}$ is a $\mathbb{Q}$-martingale.

Following the proof of Theorem $1.4$ of \cite{Sch33}, for each fixed $n$, we will only consider the $\lambda_n$-CPS $(\mathbb{Q}, \tilde{S})$ such that $\tilde{S}$ takes values in the spread $[(1-\lambda_n)S^{\tau_n}, S^{\tau_n}]$ and $\tilde{S}$ is a true $\mathbb{Q}$-martingale. Let $(Z^0,Z^1)$ denote the associated martingales with respect to the $\lambda_n$-CPS $(\mathbb{Q}, \tilde{S})$ for the stopped price process $S^{\tau_n}$. We will construct a $\lambda$-CPS $(\bar{Z}^0, \bar{Z}^1)$ for the original price process $S$. Fix $0<\lambda' <\frac{\lambda-\lambda_n}{2}$. Assumption $\ref{Assum}$ gives the existence of a $\lambda'$-CPS $(\hat{Z}^0, \hat{Z}^1)$ for $S$ where $\hat{Z}^0$ is a martingale and $\hat{Z}^1$ is a local martingale.

Let us define
\begin{equation}
\bar{Z}_t^0=\left\{
\begin{array}{rl}
Z_t^0,&\ \ \ \ 0\leq t\leq \tau_n,\\
\hat{Z}_t^0\frac{Z^0_{\tau_n}}{\hat{Z}^0_{\tau_n}},&\ \ \ \ \tau_n\leq t\leq T,
\end{array}\right.\nonumber
\end{equation}
and also
\begin{equation}
\bar{Z}_t^1=\left\{
\begin{array}{rl}
(1-\lambda')Z_t^1,&\ \ \ \ 0\leq t\leq \tau_n,\\
(1-\lambda')\hat{Z}_t^1\frac{Z^1_{\tau_n}}{\hat{Z}^1_{\tau_n}},&\ \ \ \ \tau_n\leq t\leq T.
\end{array}\right.\nonumber
\end{equation}

It is clear that $\bar{Z}^0$ (resp. $\bar{Z}^1$) is a positive martingale (resp. local martingale) under $\mathbb{P}$ and $\frac{d\bar{\mathbb{Q}}}{d\mathbb{P}}=\bar{Z}^0_T$ defined a probability measure on $\mathcal{F}$ which is equivalent to $\mathbb{P}$. Moreover, for $0\leq t\leq \tau_n$, we have $\frac{\bar{Z}_t^1}{\bar{Z}_t^0}$ stays in the spread $[(1-\lambda_n)(1-\lambda')S_t, (1-\lambda')S_t]$. On the other hand, for $\tau_n\leq t\leq T$, we can verify that $\frac{\bar{Z}_t^1}{\bar{Z}_t^0}$ lies in $[(1-\lambda_n)(1-\lambda')^2S_t, (1-\lambda')^2S_t]$. It follows that $\frac{\bar{Z}^1}{\bar{Z}^0}$ takes its values in $[(1-\lambda)S, S]$. We first claim that
\begin{equation}\label{addgn}
\mathbb{E}^{\bar{\mathbb{Q}}}[g^n]\leq \mathbb{E}^{\bar{\mathbb{Q}}}[g].
\end{equation}
To see this, let us denote $f\triangleq -\underset{\tilde{S}\in\mathcal{S}(\lambda)}{\text{essinf}}\ \Big(X_T^{\max,\tilde{S}}\Big)$. As $\bar{\mathbb{Q}}\in\mathcal{M}(\lambda, S)$, it follows that
\begin{align*}
\mathbb{E}^{\bar{\mathbb{Q}}}[g^n]&=\mathbb{E}^{\bar{\mathbb{Q}}}[g\mathbf{1}_{\{\tau_n=T\}}]+\mathbb{E}^{\bar{\mathbb{Q}}}[g^n\mathbf{1}_{\{\tau_n<T\}}]\\
&\leq \mathbb{E}^{\bar{\mathbb{Q}}}[g\mathbf{1}_{\{\tau_n=T\}}]+\mathbb{E}^{\bar{\mathbb{Q}}}[ \mathbb{E}^{\bar{\mathbb{Q}}}[ f|\mathcal{F}_{\tau_n}]\mathbf{1}_{\{\tau_n<T\}}]\\
&=\mathbb{E}^{\bar{\mathbb{Q}}}[g\mathbf{1}_{\{\tau_n=T\}}]+\mathbb{E}^{\bar{\mathbb{Q}}}[f \mathbf{1}_{\{\tau_n<T\}} ].
\end{align*}
Recall that $g\geq f$, $\mathbb{P}$-a.s., and therefore $g\mathbf{1}_{\{\tau_n<T\}}\geq f\mathbf{1}_{\{\tau_n<T\}}$, $\mathbb{P}$-a.s.. It follows that $g\mathbf{1}_{\{\tau_n<T\}}\geq f\mathbf{1}_{\{\tau_n<T\}}$, $\bar{\mathbb{Q}}$-a.s. because $\bar{\mathbb{Q}}\sim\mathbb{P}$. We deduce that $\mathbb{E}^{\bar{\mathbb{Q}}}[f \mathbf{1}_{\{\tau_n<T\}} ]\leq \mathbb{E}^{\bar{\mathbb{Q}}}[g \mathbf{1}_{\{\tau_n<T\}} ]$, which implies that $(\ref{addgn})$ holds. By $(\ref{ineassum})$, $(\ref{addgn})$ and the fact that $g^n$ is $\mathcal{F}_{\tau_n}$-measurable, we can conclude that
\begin{equation}\label{ineg}
\mathbb{E}^{\mathbb{Q}}[g^n]=\mathbb{E}^{\bar{\mathbb{Q}}}[g^n]\leq \mathbb{E}^{\bar{\mathbb{Q}}}[g]\leq x.
\end{equation}

For each fixed $n\in\mathbb{N}$, consider a pair $(\phi_T^{0,n},\phi_T^{1,n})\notin \mathcal{U}_x(\lambda_n, S^{\tau_n})$ where we consider the stopped process $S^{\tau_n}$ as the underlying price process with transaction costs $\lambda_n$ such that $\phi_t^{i,n}=\phi_{\tau_n}^{i,n}$ for $\tau_{n}\leq t\leq T$. By the definition of the set $\mathcal{W}(x;\lambda_n, S^{\tau_n})$ and the characterization $(\ref{chaUacp})$ of $\mathcal{W}(x;\lambda_n, S^{\tau_n})$, for any constant $a_n>0$ and any $\widetilde{X}^{\max,\tilde{S}^n}\in\mathcal{X}(\tilde{S}^{n}, a_n)$ with the property $(\phi_T^{0,n}, \phi_T^{1,n})\succeq (-\widetilde{X}_T^{\max,\tilde{S}^n}, 0)$, we have
\begin{align*}
\mathbb{E}[W^{0,n}Z_T^{0,n}+W^{1,n}Z_T^{1,n}]&\triangleq \mathbb{E}[(\phi_T^{0,n}+\essinf_{\tilde{S}^n\in\mathcal{S}}\widetilde{X}_T^{\max,\tilde{S}^n})Z_T^{0,n}+\phi_T^{1,n}Z_T^{1,n}]\\
&> x+\mathbb{E}[\essinf_{\tilde{S}^n\in\mathcal{S}}\widetilde{X}_T^{\max,\tilde{S}^n}Z_T^{0,n}],\nonumber
\end{align*}
for some $(Z^{0,n}_T, Z^{1,n}_T)\in\bar{\mathcal{Z}}(\lambda_n, S^{\tau_n})$.

In particular, we can choose some maximal elements $\widetilde{X}_t^{\max,\tilde{S}^n}\equiv a_n$ for $0\leq t\leq T$ and hence $\essinf_{\tilde{S}^n\in\mathcal{S}}\widetilde{X}_T^{\max,\tilde{S}^n}\leq a_n$ is integrable. It follows that $\phi_T^{0,n}+\phi_T^{1,n}\frac{Z_T^{1,n}}{Z_T^{0,n}}$ is $\mathbb{Q}^n$ integrable where $\frac{d\mathbb{Q}^n}{d\mathbb{P}}=Z_T^{0,n}$. We can obtain that
\begin{equation}
\mathbb{E}^{\mathbb{Q}^n}\left[\phi_T^{0,n}+\phi_T^{1,n}\frac{Z_T^{1,n}}{Z_T^{0,n}}\right]>x.\nonumber
\end{equation}
In the case that $\mathbb{Q}^n$ is only absolutely continuous with respect to $\mathbb{P}$, but not equivalent to $\mathbb{P}$, The above argument asserts that any $\lambda^n$-CPS for $S$ is $\lambda^n$-CPS for $S^{\tau_n}$. Therefore, there exists some $(\bar{Z}^0_T, \bar{Z}^1_T)\in\bar{\mathcal{Z}}(\lambda^n, S^{\tau_n})$ such that $\bar{Z}_T^0>0$ a.s. For $0<\beta<1$ sufficiently small, define $\hat{Z}^n=\beta\bar{Z}_T+(1-\beta)Z^{n}_T$. We obtain that $\hat{Z}_T^{0,n}>0$ a.s. Define
\begin{equation}
\hat{S}^n_t=\frac{\hat{Z}^{1,n}_t}{\hat{Z}^{0,n}_t},\ \ \ \text{and}\ \ \ \frac{d\hat{\mathbb{Q}}^n}{d\mathbb{P}}=\hat{Z}_T^{0,n}.\nonumber
\end{equation}
We have $(\hat{\mathbb{Q}}^n, \hat{S}^n)$ is $\lambda^n$-CPS and also $\mathbb{E}^{\hat{\mathbb{Q}}^n}[\phi_T^{0,n}+\phi_T^{1,n}\hat{S}^n_T]>x$, which can not satisfy $(\ref{ineg})$. Therefore, we obtain that if $(\ref{ineg})$ holds, there exits a pair $(\phi_T^{0,n}, \phi_T^{1,n})\in\mathcal{U}_x(\lambda^n, S^{\tau_n})$ such that $(\phi_0^{0,n}, \phi_0^{1,n})=(x,0)$ and $(\phi_T^{0,n}, \phi_T^{1,n})=(\phi_{\tau_n}^{0,n}, \phi_{\tau_n}^{1,n})=(g^n, 0)$.

By taking the limit $n\rightarrow\infty$ and the convex combinations of $(\phi^{0,n}, \phi^{1,n})$, similar to the proof of Lemma $\ref{fatoulem}$, we can conclude that $(\phi^0, \phi^1)\in\mathcal{U}_x(\lambda,S)$, which completes the proof.
\end{proof}

To prove Proposition $\ref{dualprop}$, we still need auxiliary results from Lemma $\ref{clK}$ to Lemma $\ref{lemma11}$ as below.

\begin{lemma}\label{clK}
Under Assumptions $\ref{Assum}$, $\ref{assE}$ and $\ref{assNo}$, we have $\bar{\mathcal{K}}=\{(x,q)\in\mathbb{R}^{1+N}:\mathcal{H}(x,q)\neq \emptyset\}$, where $\bar{\mathcal{K}}$ is the closure of the set $\mathcal{K}$ in $\mathbb{R}^{1+N}$.
\end{lemma}
\begin{proof}
Fix any $(x,q)\in\bar{\mathcal{K}}$, and let $(x^n,q^n)_{n\geq 1}$ be a sequence in $\mathcal{K}$ that converges to $(x,q)$. We need to verify that $\mathcal{H}(x,q)\neq\emptyset$. Choose a sequence $V_T^n\in\mathcal{H}(x^n,q^n)$ with $V_T^n=V(\phi^{0,n},\phi^{1,n})_T$ and $(\phi^{0,n}, \phi^{1,n})\in\mathcal{A}_{x^n}$, $n\geq 1$. Lemma $\ref{ineqC}$ gives that
\begin{equation}\label{lemineqC}
\mathbb{E}^{\mathbb{Q}}[V(\phi^{0,n}, \phi^{1,n})_T+q^n\cdot\mathcal{E}_T]\leq x^n+\mathbb{E}^{\mathbb{Q}}[q^n\cdot\mathcal{E}_T],\ \ \forall \mathbb{Q}\in\mathcal{M}.
\end{equation}
In addition, the fact that $x^n\rightarrow x$ and $q^n\rightarrow q$ imply that that there exists finite constants $k_1$ and $k_2$ such that $x^n<k_1$ and $(q^n)^i<k_2$, $1\leq i\leq N$, for $n$ large enough. We deduce that $q^n\cdot\mathcal{E}_T \leq k_2\sum_{i=1}^{N}\mathcal{E}_T^i$. By Lemma $\ref{boundE}$, it follows that there exists a constant $\hat{a}>0$ and for each $\tilde{S}\in\mathcal{S}(\lambda)$, there exists a $\hat{X}^{\max,\tilde{S}}\in\mathcal{X}(\tilde{S},\hat{a})$ such that $V(\phi^{0,n}, \phi^{1,n})_T+\hat{X}_T^{\max,\tilde{S}}\geq 0$ for $n$ large enough. Lemma $\ref{lemm}$ and Lemma $A1.1$ of \cite{Sch94} imply that we can find the convex combinations of $\phi_T^{0,n}$ and $\phi_T^{1,n}$ converging almost surely to random variables $\phi_T^{0}$ and $\phi_T^{1}$ respectively. Moreover, it is clear that $V(\phi^0,\phi^1)_T+q\cdot\mathcal{E}_T\geq 0$ a.s. where $V(\phi^0,\phi^1)_T=\phi_T^0+(\phi_T^1)^+(1-\lambda)S_T-(\phi_T)^-S_T$. Fatou's Lemma and $(\ref{lemineqC})$ therefore imply that
\begin{align*}
&\mathbb{E}^{\mathbb{Q}}[V(\phi^0,\phi^1)_T+q\cdot\mathcal{E}_T]\leq \lim_{n\rightarrow\infty}\mathbb{E}^{\mathbb{Q}}[V(\phi^{0,n},\phi^{1,n})_T+q^n\cdot\mathcal{E}_T]\\
\leq &\lim_{n\rightarrow\infty}\Big(x^n+\mathbb{E}^{\mathbb{Q}}[q^n\cdot\mathcal{E}_T]\Big)=x+\mathbb{E}^{\mathbb{Q}}[q\cdot\mathcal{E}_T],\ \ \forall \mathbb{Q}\in\mathcal{M}.
\end{align*}
It follows that $\mathbb{E}^{\mathbb{Q}}[V(\phi^0,\phi^1)_T]\leq x$, $\forall \mathbb{Q}\in\mathcal{M}$. Lemma $\ref{charac}$ guarantees the existence of acceptable portfolio $(\hat{\phi}^0,\hat{\phi}^1)\in\mathcal{A}_x$ such that $V(\hat{\phi}^0,\hat{\phi}^1)_T\geq V(\phi^0,\phi^1)_T\geq -q\cdot\mathcal{E}_T$. Therefore, we obtain that $V(\hat{\phi}^0,\hat{\phi}^1)_T\in\mathcal{H}(x,q)$.
\end{proof}

For a vector $p\in\mathbb{R}^N$, we define the set
\begin{equation}\label{auxiM}
\mathcal{M}(p)=\{\mathbb{Q}\in\mathcal{M}: \mathbb{E}^{\mathbb{Q}}[\mathcal{E}_T]=p\}
\end{equation}
From its definition, $\mathcal{P}$ is the intersection of $\mathcal{L}$ with the hyperplane $y\equiv 1$ which defines the set of arbitrage-free prices of the contingent claim $\mathcal{E}_T$.

\begin{lemma}\label{lemma9}
Assume that all conditions of Proposition $\ref{dualprop}$ hold and let $p\in\mathbb{R}^N$. The set $\mathcal{M}(p)$ is not empty if and only if $p\in\mathcal{P}$. In particular, $\bigcup_{p\in\mathcal{P}}\mathcal{M}(p)=\mathcal{M}$.
\end{lemma}
\begin{proof}
Under Assumptions $\ref{assE}$ and $\ref{assNo}$, Lemma $\ref{lemma9}$ follows directly from the proof Lemma $8$ of \cite{kram04}, if we replace the set $\mathcal{M}'(p)$, Lemma $4$, Lemma $5$ and Lemma $6$ in \cite{kram04} by the set $\mathcal{M}(p)$, Lemma $\ref{ineqC}$, Lemma $\ref{charac}$ and Lemma $\ref{clK}$ in this paper.
\end{proof}

\begin{lemma}\label{lemma10}
Under the assumptions of Proposition $\ref{dualprop}$ and $p\in\mathcal{P}$, the density process of any $\mathbb{Q}\in\mathcal{M}(p)$ belongs to $\mathcal{Y}(1,p)$.
\end{lemma}
\begin{proof}
According to the definition of CPS and Proposition $2.3$ of \cite{Sch2}, it is clear that the density process of $\mathbb{Q}\in\mathcal{M}$ belongs to $\mathcal{Y}(1)$ defined by $(\ref{dualY})$ . The conclusion follows by Lemma $\ref{ineqC}$ and the definition of $\mathcal{M}(p)$.
\end{proof}

\begin{lemma}\label{lemma11}
Under the assumptions of Proposition $\ref{dualprop}$, a nonnegative random variable $g$ belongs to $\mathcal{C}(x,q)$ where $(x,q)\in\mathcal{K}$ if and only if
\begin{equation}\label{needto}
\mathbb{E}^{\mathbb{Q}}[g]\leq x+p\cdot q,\ \ \forall p\in\mathcal{P}\ \text{and}\ \mathbb{Q}\in\mathcal{M}(p).
\end{equation}
\end{lemma}
\begin{proof}
Suppose $g\in\mathcal{C}(x,q)$, Lemma $\ref{ineqC}$ implies the inequality $(\ref{needto})$. On the other hand, consider the random variable $\beta\triangleq g-q\cdot\mathcal{E}_T$. It follows from $(\ref{needto})$ that
\begin{align*}
\sup_{\mathbb{Q}\in\mathcal{M}}\mathbb{E}^{\mathbb{Q}}[\beta]&=\sup_{p\in\mathcal{P}}\sup_{\mathbb{Q}\in\mathcal{M}(p)}\mathbb{E}^{\mathbb{Q}}[\beta]=\sup_{p\in\mathcal{P}}\sup_{\mathbb{Q}\in\mathcal{M}(p)}(\mathbb{E}^{\mathbb{Q}}[g]-q\cdot p)\leq x.\nonumber
\end{align*}
Assumption $\ref{assE}$ and Lemma $\ref{boundE}$ imply the existence of a constant $\hat{a}>0$ such that for each $\tilde{S}\in\mathcal{S}$, there exists a $\hat{X}^{\max,\tilde{S}}\in\mathcal{X}(\tilde{S},\hat{a})$ and $\beta\geq -\hat{X}_T^{\max,\tilde{S}}$. Lemma $\ref{charac}$ guarantees the existence of acceptable portfolio with $\phi^0_0=x$, $\phi_0^1=0$ and $V(\phi^0,\phi^1)_T=\phi_T^0\geq \beta$. We therefore obtain that
\begin{equation}
0\leq g\leq V(\phi^0,\phi^1)_T+q\cdot\mathcal{E}_T,\nonumber
\end{equation}
which implies that $V(\phi^0,\phi^1)_T\in\mathcal{H}(x,q)$ and $g$ belongs to $\mathcal{C}(x,q)$.
\end{proof}

\begin{proof}[Proof of Proposition $\ref{dualprop}$] We first prove the assertion $(i)$. Assume that $(x,q)\in\mathcal{K}$. We can find a constant $\delta>0$ such that $(x-\delta, q)\in\mathcal{K}$ since $\mathcal{K}$ is open. Consider $V_T=V(\phi^0,\phi^1)_T\in\mathcal{H}(x-\delta,q)$, it is clear that $\widetilde{V}\triangleq V(\phi^0+\delta,\phi^1)$ is in $\mathcal{H}(x,q)$ and $\delta\leq \widetilde{V}_T+q\cdot\mathcal{E}_T$ which implies that $\delta\in\mathcal{C}(x,q)$.

Let $(x,q)\in\mathcal{K}$. If $g\in\mathcal{C}(x,q)$, $(\ref{bipo1})$ holds true by the definition of $\mathcal{D}(y,r)$, $(y,r)\in\mathcal{L}$. On the other hand, consider a nonnegative random variable such that $(\ref{bipo1})$ holds. It follows that $g$ satisfies $(\ref{needto})$ by Lemma $\ref{lemma10}$.  Lemma $\ref{lemma11}$ then implies that $g$ belongs to $\mathcal{C}(x,q)$.

It is clear that $k\mathcal{D}(y,r)=\mathcal{D}(ky,kr)$ for any $k>0$ and $(y,r)\in\mathcal{L}$. Hence, to verify the assertion $(ii)$, it is enough to consider the case that $(y,r)=(1,p)$ for some $p\in\mathcal{P}$. Due to Lemma $\ref{lemma10}$, there exists a process $Y_t=\mathbb{E}\Big[\frac{d\mathbb{Q}}{d\mathbb{P}}\Big|\mathcal{F}_t\Big]$ with $\mathbb{Q}\in\mathcal{M}(p)$, which satisfies $Y_T\in\mathcal{D}(1,p)$ and $Y_T>0$ a.s. For any $h\in\mathcal{D}(1,p)$, $(\ref{bipo2})$ holds by the definition of $\mathcal{D}(1,p)$. Conversely, consider any nonnegative random variable $h$ satisfying $(\ref{bipo2})$. In particular, we have that $\mathbb{E}^{\mathbb{Q}}[gh]\leq 1$, $\forall g\in\mathcal{C}(1,0)$. Because $\mathcal{C}(1,0)=\mathcal{V}_1^{\text{adm}}$ which is defined in $(\ref{0admV})$, Lemma $A.1$ in \cite{Chris22} asserts the existence of an optional strong supermartingale $(Y^0, Y^1)\in\mathcal{Z}(1)$ such that $h\leq Y_T^0$. Let us define the process $\tilde{Y}$ by
\begin{equation}
\tilde{Y}_t=\left\{
\begin{array}{rl}
Y_t^0,&\ \ \ \ t<T,\\
h,&\ \ \ \ t=T.
\end{array}\right.\nonumber
\end{equation}
It follows that $\tilde{Y}\in\mathcal{Y}(1,p)$. Therefore we obtain that $h\in\mathcal{D}(1,p)$.
\end{proof}

\begin{proof}[Proof of Theorem $\ref{mainthm}$]
Once we build the bipolar results in Proposition $\ref{dualprop}$, Theorem $\ref{mainthm}$ follows the proof of Theorem $1$ and Theorem $2$ of \cite{kram04} if we replace the one-dimensional duality theory in \cite{kram99} by Theorem $3.2$ in \cite{Chris22} under proportional transaction costs.
\end{proof}

\subsection{Proof of Theorem $\ref{theorem1}$}
\begin{proof}
By Theorem $4$ in Appendix I in \cite{Meyer}, every optional strong supermartingale is indistinguishable from a \ladlag process. Without loss of generality, we can assume all optional strong supermartingales are \ladlag. In particular, we can assume that $\hat{S}=\frac{Y^{1,\ast}}{Y^{0,\ast}}$ is \ladlag. Fix $(x,q)\in\mathcal{K}$, for any $(y,r)\in\partial u(x,q)$, by the self-financing condition and integration by parts, we deduce that
\begin{align*}
Y_t^{0,\ast}(y, r)\phi_t^{0,\ast}(x,q)+Y^{1,\ast}_t(y,r)\phi_t^{1,\ast}(x,q)=&Y^{0,\ast}_t(y,r)(\phi_t^{0,\ast}(x,q)+\phi_t^{1,\ast}(x,q)\hat{S}_t)\\
=&Y^{0,\ast}_t(y,r)(x+(\phi^{1,\ast}\cdot \hat{\mathbf{S}})_t+K_t)\nonumber
\end{align*}
where $(K_t)_{0\leq t\leq T}$ is a non-increasing predictable process defined by
\begin{align*}
K_t\triangleq &\int_0^t(\hat{S}_u-S_u)d\phi^{1,\ast, \uparrow, c}_u(x,q)+ \int_0^1((1-\lambda)S_u-\hat{S}_u)d\phi_u^{1,\ast,\downarrow, c}(x,q)\\
&+\sum_{0<u\leq t}(\hat{S}_{u}^p-S_{u-})\triangle \phi_u^{1,\ast, \uparrow}(x,q)+\sum_{0<u\leq t}((1-\lambda)S_{u-}-\hat{S}_u^p)\triangle \phi_u^{1,\ast,\downarrow}(x,q)\\
&+\sum_{0\leq u<t}(\hat{S}_u- S_u)\triangle_+\phi_u^{1,\ast,\uparrow}(x,q)+\sum_{0\leq u<t}(\hat{S}_u-(1-\lambda)S_u)\triangle_+\phi^{1,\ast,\downarrow}_u(x,q)\nonumber
\end{align*}
for $t\in[0,T]$. Therefore, to show that $(\ref{equiphi})$ holds is equivalent to show that $(\ref{behavphi})$ holds.

Under Assumption $\ref{importantnew}$, for some $(y,r)\in\partial u(x,q)$, there exists a minimizing sequence $Z^n(y, r)$ in $\mathcal{B}(1)$ such that $\underset{n\rightarrow\infty}{\lim\inf}\mathbb{E}[Z_T^{0,n}(y, r)\mathcal{E}_T]=\frac{r}{y}$. By defining $\tilde{S}^n\triangleq\frac{Z^{1,n}(y, r)}{Z^{0,n}(y, r)}$, we can see that $\tilde{S}^n$ stays in the bid-ask spread $[(1-\lambda)S, S]$ and $\tilde{S}^n\in\mathcal{S}$ under the transaction costs $\lambda$. Using the integration by parts formula again, we get
\begin{align*}
\phi_t^{0,\ast}(x,q)+\phi_t^{1,\ast}(x,q)\tilde{S}_t^n=&\phi_t^{0,\ast}(x,q)+\int_0^t\phi_u^{1,\ast}(x,q)d\tilde{S}^n_u+\int_0^t\tilde{S}_u^{n}d\phi_u^{1,\ast,c}(x,q)\\
&+\sum_{0<u\leq t}\tilde{S}_{u-}^n\triangle \phi_u^{1,\ast}(x,q)+\sum_{0\leq u<t}\tilde{S}_u^{n}\triangle_+\phi_u^{1,\ast}(x,q),
\end{align*}
so that we can write
\begin{equation}
\phi_t^{0,\ast}(x,q)+\phi_t^{1,\ast}(x,q)\tilde{S}_t^n=x+\int_0^t\phi_u^{1,\ast}(x,q)d\tilde{S}_u^n+K_t^n,\nonumber
\end{equation}
where
\begin{align*}
K_t^n\triangleq &\int_0^t(\tilde{S}_u^n-S_u)d\phi_u^{1,\ast,\uparrow, c}(x,q)+\int_0^t((1-\lambda)S_u-\tilde{S}_u^n)d\phi_u^{1,\ast,\downarrow, c}(x,q)\\
&+\sum_{0<u\leq t}(\tilde{S}_{u-}^{n}-S_{u-})\triangle \phi_u^{1,\ast,\uparrow}(x,q)+\sum_{0<u\leq t}((1-\lambda)S_{u-}-\tilde{S}_{u-}^{n})\triangle \phi_u^{1,\ast,\downarrow}(x,q)\\
&+\sum_{0\leq u< t}(\tilde{S}_{u}^{n}-S_{u})\triangle_+ \phi_u^{1,\ast,\uparrow}(x,q)+\sum_{0\leq u< t}((1-\lambda)S_{u}-\tilde{S}_{u}^{n})\triangle_+ \phi_u^{1,\ast,\downarrow}(x,q).
\end{align*}
is a non-increasing predictable process.

As $\phi^{1,\ast}(x,q)$ is predictable and of finite variation, it is clear from integration by parts that $Z^{0,n}(y, r)(x+\phi^{1,\ast}(x,q)\cdot \tilde{S}^n)$ is a local martingale. For the choice of $\tilde{S}^n\in\mathcal{S}$, by the definition of acceptable portfolio, there exists a maximal element $X^{\max,\tilde{S}^n}$ such that
\begin{equation}
x+\int_0^t\phi_u^{1,\ast}(x,q)d\tilde{S}^n_u+X_t^{\max,\tilde{S}^n}\geq V(\phi^{0,\ast}(x,q), \phi^{1,\ast}(x,q))_t+X_t^{\max, \tilde{S}^n}\geq 0.\nonumber
\end{equation}
Also, denote the measure $\frac{d\mathbb{Q}^n}{d\mathbb{P}}=Z_T^{0,n}(y, r)$, we have $\mathbb{Q}^n\in\mathcal{M}(\tilde{S}^n)$. Consider the subset
\begin{equation}
\mathcal{M}'(\tilde{S}^n)\triangleq \{\mathbb{Q}\in\mathcal{M}(\tilde{S}^n):\ \text{$X^{\max,\tilde{S}^n}$ is a UI martingale under $\mathbb{Q}$}\}.\nonumber
\end{equation}
There exists a sequence $(\mathbb{Q}^{n,m})_{m=1}^{\infty}$ in $\mathcal{M}'(\tilde{S}^n)$ converging to $\mathbb{Q}^n$ in the norm topology of $\mathbb{L}^1(\Omega, \mathcal{F}, \mathbb{P})$. For each $\mathbb{Q}^{n,m}\in\mathcal{M}'(\tilde{S}^n)$, it follows that $(x+\phi^{1,\ast}(x,q)\cdot \tilde{S}^n+X^{\max,\tilde{S}^n})$ is a true supermartingale under $\mathbb{Q}^{n,m}$. Hence, we can derive that
\begin{align*}
&\mathbb{E}^{\mathbb{Q}^{n,m}}\Big[x+\int_0^T\phi_u^{1,\ast}(x,q)d\tilde{S}^n_u+q\mathcal{E}_T\Big]\\
=&\mathbb{E}^{\mathbb{Q}^{n,m}}\Big[x+\int_0^T\phi_u^{1,\ast}(x,q)d\tilde{S}^n_u+X_T^{\max,\tilde{S}^n}\Big] -\mathbb{E}^{\mathbb{Q}^{n,m}}[X_T^{\max,\tilde{S}^n}] +\mathbb{E}^{\mathbb{Q}^{n,m}}[q\mathcal{E}_T]\\
\leq &x+\mathbb{E}^{\mathbb{Q}^{n,m}}[q\mathcal{E}_T].\nonumber
\end{align*}

Following the proof of Lemma $\ref{ineqC}$ and by passing the limit as $m\rightarrow \infty$, we obtain that
\begin{equation}\label{addlate1}
\mathbb{E}\Big[Z^{0,n}_T(y, r) (x+\int_0^T \phi_u^{1,\ast}(x,q)d\tilde{S}_u^n+q\mathcal{E}_T) \Big]\leq x+\mathbb{E}[Z_T^{0,n}(y, r)q\mathcal{E}_T].
\end{equation}
By Fatou's lemma, Lemma $\ref{seplemma}$ and $(\ref{addlate1})$, we obtain that
\begin{align*}
xy+qr=&\mathbb{E}[Y^{0,\ast}_T(y, r)(\phi_T^{0,\ast}(x,q)+q\mathcal{E}_T)]\leq \underset{n\rightarrow\infty}{\lim\inf}\mathbb{E}[yZ_T^{0,n}(y, r)(\phi_T^{0,\ast}(x,q)+q\mathcal{E}_T) ]\\
\leq& \underset{n\rightarrow\infty}{\lim\inf}\mathbb{E}[yZ_T^{0,n}(y, r)K_T^n]+xy+\underset{n\rightarrow\infty}{\lim\inf}\mathbb{E}[yZ_T^{0,n}(y, r)q\mathcal{E}_T]\\
=& \underset{n\rightarrow\infty}{\lim\inf}\mathbb{E}[yZ_T^{0,n}(y, r)K_T^n]+xy+qr.\nonumber
\end{align*}

Therefore, it holds that $Z_T^{0,n}(y, r)K_T^n$ converges to $0$ in $\mathbb{L}^1(\mathbb{P})$ as $K_T^n\leq 0$. We can mimic the proof of Theorem $3.5$ of \cite{Chris22} and show that $K_T^n$ converges to $K_T$ almost surely, and hence $K_T=0$. As $K_0=0$ and $K_t$ is a non-increasing process, $(\ref{behavphi})$ is verified and hence $(\ref{equiphi})$ also holds true.
\end{proof}

\subsection{Proof of Theorem $\ref{theosandsha}$}
\begin{proof}
Under all assumptions of Theorem $\ref{theorem1}$, for some $(y,r)\in\partial u(x,q)$ in Assumption $\ref{importantnew}$, let $Z^n(y, r)$ be the minimizing sequence which satisfies $(\ref{marginalprice})$, $(\ref{con111})$, $(\ref{con222})$. For any $X(\phi^0,\phi^1)_T\in\mathcal{H}(x,q;\hat{\mathbf{S}})$, using Definition $\ref{sandSaccpt}$ of acceptable portfolios under the sandwiched shadow price $\hat{\mathbf{S}}$, we deduce that
\begin{equation}
\phi_T^0+\phi_T^1\tilde{S}_T^n+q\mathcal{E}_T=\phi_T^0+q\mathcal{E}_T\geq V(\phi^0, \phi^1)_T+q\mathcal{E}_T\geq 0.\nonumber
\end{equation}
where $\tilde{S}^n\triangleq \frac{Z^{1,n} (y,r)}{Z^{0,n}(y,r)}\in\mathcal{S}$ under transaction costs $\lambda$. Fatou's lemma implies that
\begin{align}\label{newineq1}
\mathbb{E}\Big[Y_T^{0,\ast}(y, r)(x+\int_0^T\phi_u^{1}d\hat{\mathbf{S}}_u+q\mathcal{E}_T)\Big]&=\mathbb{E}[Y_T^{0,\ast}(y, r)(\phi_T^{0}+\phi_T^{1}\hat{S}_T +q\mathcal{E}_T)]\nonumber\\
&\leq \underset{n\rightarrow\infty}{\lim\inf}\mathbb{E}[yZ_T^{0,n} (y,r)(\phi_T^0+q\mathcal{E}_T)].
\end{align}

Again, Definition $\ref{sandSaccpt}$ gives the existence of some $X^{\max,\tilde{S}^n}\in\mathcal{X}(\tilde{S}^n, a)$ for some constant $a>0$ such that $\phi_T^0+X_T^{\max,\tilde{S}^n}\geq 0$. By the similar proof of Theorem $\ref{theorem1}$ above, we deduce that
\begin{equation}\label{newineq2}
\lim_{n\rightarrow\infty}\mathbb{E}[yZ_T^{0,n} (y,r)(\phi_T^0+q\mathcal{E}_T)]\leq xy+qr=\mathbb{E}[Y^{0,\ast}_T(y, r)(\phi_T^{0,\ast}(x,q)+q\mathcal{E}_T)].
\end{equation}

Fenchel's inequality implies that
\begin{align*}
&\mathbb{E}\Big[U(x+\int_0^T\phi_u^1d\hat{\mathbf{S}}_u+q\mathcal{E}_T)\Big]\\
\leq &\mathbb{E}\Big[\tilde{U}(Y_T^{0,\ast}(y, r))+Y_T^{0,\ast}(y, r)(x+\int_0^T\phi_u^1d\hat{\mathbf{S}}_u+q\mathcal{E}_T)\Big]\\
=&\mathbb{E}[\tilde{U}(Y_T^{0,\ast}(y, r))+Y_T^{0,\ast}(y, r)(\phi_T^0+\phi_T^1\hat{S}_T+q\mathcal{E}_T)]\\
\leq &\mathbb{E}[\tilde{U}(Y_T^{0,\ast}(y,r))+Y_T^{0,\ast}(y,r)(\phi_T^{0,\ast}(x,q)+\phi_T^{1,\ast}(x,q)\hat{S}_T+q\mathcal{E}_T)]
\end{align*}
using $(\ref{newineq1})$ and $(\ref{newineq2})$. Therefore, by $(\ref{ok1})$, we can verify that
\begin{align*}
&\mathbb{E}[U(X(\phi^0,\phi^1)_T+q\mathcal{E}_T)]\\
=&\mathbb{E}\Big[U(x+\int_0^T\phi_u^1d\hat{\mathbf{S}}_u+q\mathcal{E}_T)\Big]\\
\leq &\mathbb{E}[\tilde{U}(Y_T^{0,\ast}(y, r))+Y_T^{0,\ast}(y, r)(\phi_T^{0,\ast}(x,q)+\phi_T^{1,\ast}(x,q)\hat{S}_T+q\mathcal{E}_T)]\\
=&\mathbb{E}[U(\phi_T^{0,\ast}(x,q)+\phi_T^{1,\ast}(x,q)\hat{S}_T+q\mathcal{E}_T)]\\
=&\mathbb{E}[U(V(\phi^{0,\ast}(x,q), \phi^{1,\ast}(x,q))_T+q\mathcal{E}_T)].
\end{align*}
\end{proof}

\subsection{Proof of Theorem $\ref{classical1}$}
\begin{proof}
Fix $(x,q)\in\mathcal{K}$ and let us consider some $(y,r)\in\partial u(x,q)$. Suppose that $(Y^{0,\ast}(y,r), Y^{1,\ast}(y,r))\in\mathcal{B}(y)$ and $Y^{0,\ast}(y,r)\in y\mathcal{M}(\frac{r}{y})$. The process $(Y^{0,\ast,p}(y,r), Y^{1,\ast,p}(y,r))$ coincides with $(Y^{0,\ast}(y,r), Y^{1,\ast}(y,r))$ and $\hat{S}\triangleq \frac{Y^{1,\ast}(y,r)}{Y^{0,\ast}(y,r)}\in\mathcal{S}$ under transaction costs, moreover, $x+\int_0^t\phi_u^{1,\ast}(x,q)d\hat{\mathbf{S}}_u=x+\int_0^t\phi_u^{1,\ast}(x,q)d\hat{S}_u$.

We claim that $Y^{0,\ast}(y,r)\in\mathcal{Y}(y,r;\hat{S})$. The proof of Proposition $3.7$ of \cite{Chris22} already asserts that $Y^{0,\ast}(y,r)\in\mathcal{Y}(y;\hat{S})$ and it is enough to verify that for any $X_T\in\mathcal{H}(x,q;\hat{S})$ and $(x,q)\in\mathcal{K}(\hat{S})$,
\begin{equation}
\mathbb{E}[Y_T^{0,\ast}(y,r)(X_T+q \mathcal{E}_T)]\leq xy+q r\nonumber
\end{equation}
As $X(\phi^0, \phi^1)_T\in\mathcal{V}_x(\hat{S})$ for some $(\phi^0, \phi^1)\in\mathcal{A}_x(\hat{S})$, we obtain that
\begin{equation}
X_T=x+\int_0^T\phi_u^1d\hat{S}_u=X_T'-X_T^{\max},\ \ \text{where}\ X',X^{\max}\in\mathcal{X}(\tilde{S}).\nonumber
\end{equation}

Consider the set $\mathcal{M}'(\hat{S})\triangleq \{\mathbb{Q}\in\mathcal{M}(\hat{S}):\ \text{$X^{\max}$ is a UI martingale under $\mathbb{Q}$}\}$.
We have that $(x+\int_0^t\phi_u^1d\hat{S}_u+X_t^{\max})_{0\leq t\leq T}$ is a nonnegative supermartingale under each $\mathbb{Q}\in\mathcal{M}'(\hat{S})$. Similar to the proof of Lemma $\ref{ineqC}$, we can choose a sequence $\mathbb{Q}^n$ converging to $\mathbb{Q}$ in the norm topology where $\frac{d\mathbb{Q}}{d\mathbb{P}}\triangleq \frac{1}{y}Y_T^{0,\ast}(y,r)$. By passing to the limit as $n\rightarrow\infty$ and under Assumption $\ref{assE}$, it yields that
\begin{equation}
\mathbb{E}\Big[Y_T^{0,\ast}(y,r)(x+\int_0^T\phi_u^1d\hat{S}_u+q\mathcal{E}_T)\Big]\leq xy+\mathbb{E}[Y_T^{0,\ast}(y,r)q\mathcal{E}_T]=xy+qr,\nonumber
\end{equation}
as $Y^{0,\ast}(y,r)\in y\mathcal{M}(\frac{r}{y})$. Therefore, the claim $Y^{0,\ast}(y,r)\in\mathcal{Y}(y,r;\hat{S})$ holds.

Fix $(x,q)\in\mathcal{K}$ and consider $(y,r)\in\partial u(x,q)$. It is easy to see that
\begin{align*}
\mathbb{E}[\tilde{U}(Y_T^{0,\ast}(y,r))]+xy+q r=&v(y,r)+xy+q r=u(x,q)\leq u(x,q;\hat{S})\\
\leq &v(y,r;\hat{S})+xy+q r\leq \mathbb{E}[\tilde{U}(Y_T^{0,\ast}(y,r))]+xy+q r
\end{align*}
because $Y_T^{0,\ast}(y,r)\in\mathcal{Y}(y,r)$. Therefore, we obtain that $u(x,q)=u(x,q;\hat{S})$ together with $(y,r)\in\partial u(x,q;\hat{S})$ and $Y_T^{0,\ast}(y,r)$ is the optimal solution to $v(y,r;\hat{S})$ defined by $(\ref{shadowdualv})$. As a consequence, $(\phi^{0,\ast}(x,q), \phi^{1,\ast}(x,q))$ is the optimal solution to the utility maximization problem $(\ref{frictionless})$ in the market $\hat{S}$ and $\hat{S}$ is a classic shadow price process.
\end{proof}

\begin{acknowledgements}
E. Bayraktar is supported in part by the National Science Foundation under grant DMS-1613170 and the Susan M. Smith Professorship. X. Yu is supported by the Hong Kong Early Career Scheme under grant 25302116 and the Start-Up Fund of the Hong Kong Polytechnic University under grant 1-ZE5A.
\end{acknowledgements}

\bibliographystyle{spmpsci}      
\bibliography{ReferenceEndowment}   

%
%

\end{document}